\documentclass[11pt]{article}
\usepackage[dvips,letterpaper,margin=1.0in]{geometry}
\usepackage{sty}

\title{Overcomplete Tensor Decomposition via Koszul--Young Flattenings}
\author{}
\date{}

\usepackage{authblk}
\author[1]{Pravesh K.\ Kothari\thanks{Email: \textit{kothari@cs.princeton.edu}. Partially supported by NSF CAREER Award \#2047933, Alfred P. Sloan Fellowship and a Google Research Scholar Award.
}}
\author[2,3]{Ankur Moitra\thanks{Email: \textit{moitra@mit.edu}. Partially supported by a Microsoft Trustworthy AI Grant, an ONR grant and a David and Lucile Packard Fellowship.}}
\author[4]{Alexander S.\ Wein\thanks{Email: \textit{aswein@ucdavis.edu}. Partially supported by an Alfred P.\ Sloan Research Fellowship and NSF CAREER Award CCF-2338091.}}

\affil[1]{Computer Science Department, Princeton}
\affil[2]{Department of Mathematics, MIT}
\affil[3]{Computer Science and Artificial Intelligence Lab, MIT}
\affil[4]{Department of Mathematics, UC Davis}

\begin{document}

\maketitle

\begin{abstract}
Motivated by connections between algebraic complexity lower bounds and tensor decompositions, we investigate Koszul--Young flattenings, which are the main ingredient in recent lower bounds for matrix multiplication. Based on this tool we give a new algorithm for decomposing an $n_1 \times n_2 \times n_3$ tensor as the sum of a minimal number of rank-1 terms, and certifying uniqueness of this decomposition. For $n_1 \le n_2 \le n_3$ with $n_1 \to \infty$ and $n_3/n_2 = O(1)$, our algorithm is guaranteed to succeed when the tensor rank is bounded by $r \le (1-\epsilon)(n_2 + n_3)$ for an arbitrary $\epsilon > 0$, provided the tensor components are generically chosen. For any fixed $\epsilon$, the runtime is polynomial in $n_3$. When $n_2 = n_3 = n$, our condition on the rank gives a factor-of-2 improvement over the classical simultaneous diagonalization algorithm, which requires $r \le n$, and also improves on the recent algorithm of Koiran (2024) which requires $r \le 4n/3$. It also improves on the PhD thesis of Persu (2018) which solves rank detection for $r \leq 3n/2$.

We complement our upper bounds by showing limitations, in particular that no flattening of the style we consider can surpass rank $n_2 + n_3$. Furthermore, for $n \times n \times n$ tensors, we show that an even more general class of degree-$d$ polynomial flattenings cannot surpass rank $Cn$ for a constant $C = C(d)$. This suggests that for tensor decompositions, the case of generic components may be fundamentally harder than that of random components, where efficient decomposition is possible even in highly overcomplete settings.
\end{abstract}

\newpage

\tableofcontents

\newpage

\section{Introduction}

A \emph{tensor} is a multi-way array $T \in \RR^{n_1 \times n_2 \times \cdots \times n_k}$ with entries denoted $T_{i_1,i_2,\ldots,i_k}$ for $i_j \in [n_j] := \{1,\ldots,n_j\}$, where $k$ is called the \emph{order} of the tensor. For vectors $a^{(1)} \in \RR^{n_1}, \, a^{(2)} \in \RR^{n_2}, \ldots, \, a^{(k)} \in \RR^{n_k}$, the associated \emph{rank-1} tensor is denoted $T = a^{(1)} \otimes a^{(2)} \otimes \cdots \otimes a^{(k)}$ and has entries $T_{i_1,i_2,\ldots,i_k} = a^{(1)}_{i_1} a^{(2)}_{i_2} \cdots a^{(k)}_{i_k}$. Also, $a^{\otimes k}$ is shorthand for $a \otimes a \otimes \cdots \otimes a$ ($k$ times). While various notions of tensor rank exist, we consider the \emph{CP (canonical polyadic)} decomposition, that is, the rank of a tensor $T$ is the minimum $r$ such that $T$ can be expressed as the sum of $r$ rank-1 tensors. In the matrix case $k=2$, this coincides with the usual notion of matrix rank.

\emph{Tensor decomposition} is the algorithmic task of breaking down a given rank-$r$ tensor into its $r$ constituent rank-1 summands. For matrices, this is ill-posed due to the ``rotation problem'' (see~\cite{moitra-book}). In contrast, minimum-rank tensor decompositions for $k \ge 3$ tend to be unique, provided the rank is small enough. This uniqueness is key to the usefulness of tensor methods in various applications throughout statistics, data science, and machine learning. For instance, minimum rank decompositions of the third/fourth \emph{moment} tensor can be used to learn the parameters of expressive statistical models, particularly ones with latent variables. We refer the reader to~\cite{kolda-survey,moitra-book,tensor-ml-survey,aravindan-chapter} for more background and applications of tensor decompositions.

In this work, we study tensor decompositions for an $n_1 \times \cdots \times n_k$ tensor $T$ of rank $r$. Any such tensor can be written as
\begin{equation}\label{eq:intro-decomp}
T = \sum_{\ell=1}^r a^{(1,\ell)} \otimes \cdots \otimes a^{(k,\ell)}.
\end{equation}
The vectors $a^{(i,\ell)} \in \RR^{n_i}$ are called the \emph{components} of the tensor $T$. In this work, we focus on the well-studied setting of ``large'' dimensional tensors of a fixed order; that is, we imagine a sequence of problem sizes where $k$ remains fixed but some or all of the dimensions $n_1,\ldots,n_k$ tend to infinity. 

Three natural and interrelated questions of interest arise when studying decompositions for $T$:
\begin{itemize}
\item {\bf Uniqueness Theorem}: A uniqueness theorem specifies an explicit list of conditions (ideally, efficiently checkable) on the collection of components $\{a^{(i,\ell)} \,:\, i \in [k], \, \ell \in [r]\}$ that guarantees the uniqueness of the decomposition~\eqref{eq:intro-decomp} up to natural ambiguities in ordering and scaling (see Definition~\ref{def:unique}). E.g., note that we cannot hope to resolve the norm of each component $a^{(i,\ell)}$, since $a \otimes b \otimes c = (2a) \otimes (b/2) \otimes c$.
\item {\bf Rank Detection}: Given a tensor $T$ of the form~\eqref{eq:intro-decomp} generated according to some model, the goal is to output $r$. Our lower bounds (i.e., impossibility results) will apply to this goal, which is, in general, easier than computing a minimum rank decomposition for $T$. 
\item {\bf Decomposition}: Given a tensor $T$ of the form~\eqref{eq:intro-decomp} generated according to some model, the goal is to recover the components $a^{(i,\ell)}$, again up to ordering and scaling.
\end{itemize}
Tensors pose a number of computational challenges because basic operations such as computing the rank or finding the best rank-1 approximation are NP-hard for tensors of order 3 and up~\cite{hastad-rank,hillar-lim,rank-hard-approx}.
As a result, we do not hope for an efficient (i.e., polynomial-time) tensor decomposition algorithm that works on all possible inputs, and it will be important to impose some non-degeneracy assumptions on the input tensor. Some common assumptions include:
\begin{itemize}
\item {\bf Generic components}: Here, the components are assumed to be chosen ``generically'' in the sense of algebraic geometry (see Section~\ref{sec:genericity}). This amounts to asking the algorithm to succeed for ``almost all'' choices of components (all but a measure-zero set).
\item {\bf Smoothed analysis model}: Here, the components are obtained by starting with an arbitrary collection of components and then perturbing them by a small amount of random noise. Additional noise may also be added to the observed tensor. While results in the generic setting tend to assume exact access to the input tensor and exact computation of linear-algebraic primitives, results in the smoothed setting also need to address the algorithm's numerical stability and robustness to noise.
\item {\bf Random components}: Here, the components are assumed to be drawn independently from a particular prior, such as i.i.d.\ Gaussian or uniform on the unit sphere. As compared to generic components, this induces rather strong assumptions on the components (e.g., near-orthogonality).
\end{itemize}
Additionally, some prior work assumes the rank-1 terms in the decomposition are \emph{symmetric} tensors $(a^{(\ell)})^{\otimes k}$, in place of the non-symmetric setup we have presented above. We will be interested in the \emph{overcomplete} setting, which in the ``square'' case $n_1 = n_2 = \cdots = n_k = n$ means $r > n$. 
While much is known about overcomplete third-order tensor decomposition in the random model, there are wide gaps in our understanding in the smoothed and generic settings. This will be our main focus in this paper. 

\subsection{Prior Work}

\paragraph{Generic Case.}

First we give an overview of prior work on third-order tensor decomposition, focusing on square tensors, i.e.\ those of format $n \times n \times n$, with $n$ large. A first question is that of identifiability: when are minimum-rank tensor decompositions unique? A tensor of format $n \times n \times n$ cannot have rank larger than $n^{2}$, because it is comprised of $n$ matrix ``slices'' and each is $n \times n$ with rank $\le n$. This means if $T$ is generated according to~\eqref{eq:intro-decomp} with $r > n^2$, the planted decomposition will not be the minimum-rank one, so we do not hope to find it. On the other hand, if $r \le cn^{2}$ for a particular constant $c \in (0,1)$ and the components are generically chosen, then the planted decomposition will be the unique rank-$r$ decomposition~\cite{identifiability}.  See also~\cite{specific-id} for identifiability results for tensors of specific sizes.

But what about computationally efficient algorithms?
Known polynomial-time algorithms require $r$ to be significantly below the identifiability threshold. For order-3 tensors, a classical method called \emph{simultaneous diagonalization} or \emph{Jennrich's algorithm} can decompose generic tensors of rank $r \le n$ (the ``undercomplete'' case)~\cite{harshman,LRA-3-way} (see~\cite{moitra-book}). Furthermore, the analysis comes with a uniqueness theorem: if $T = \sum_{\ell=1}^r a^{(\ell)} \otimes b^{(\ell)} \otimes c^{(\ell)}$ where $\{a^{(\ell)}\}$ and $\{b^{(\ell)}\}$ are each linearly independent and $\{c^{(\ell)}\}$ are pairwise linearly independent, then the decomposition is unique and the algorithm will find it. Another well-known uniqueness theorem due to Kruskal~\cite{kruskal} has the advantage of being applicable even for rank up to $r \le 3n/2$ but the disadvantage that the conditions on the components are not efficiently checkable and no algorithm is known for finding the decomposition in polynomial time; see also~\cite{lovitz-kruskal,gubkin-unique} and references within for various extensions of Kruskal's theorem.

\paragraph{Random Case.}

If the components are taken to be random rather than generic, there are efficient algorithms that succeed for much larger values of $r$. Known algorithms reach $r \approx n^{3/2}$ (up to logarithmic factors in $n$)~\cite{GM-sos,MSS,fast-sos,kivva-exact,ding-fast}. However these algorithms exploit the randomness of the components in an essential way:  They rely on the observation that the maxima of $\langle T, x \otimes y \otimes z \rangle$ over unit vectors $x$, $y$ and $z$ occur near the unknown components. Thus many algorithms for decomposing random third-order tensors start from a semidefinite programming relaxation for the above polynomial optimization problem and show how to round its solutions to recover the components. However this strategy quickly breaks down even in the smoothed model because the magnitude of the perturbations can be smaller than the magnitude of the components in the base tensor by a large polynomial factor. 

It appears likely that this threshold --- and more generally, $n^{k/2}$ for order-$k$ tensors --- is the fundamental limit for efficient algorithms in the random model, as evidenced by lower bounds against low-degree polynomials~\cite{tensor-decomp-LD}. We remark that for fourth-order tensors (discussed in Section~\ref{sec:further}), the algorithms for the generic model match the lower bounds for the (easier) random model, meaning we (likely) cannot expect further improvements in the generic setting, and the generic setting is no harder than the random setting. In contrast, known algorithms for third-order tensors have a gap between the generic ($r \approx n$) and random ($r \approx n^{3/2}$) settings.

\paragraph{Generic third-order overcomplete tensor decomposition.}

Our main focus will be on algorithms for decomposing overcomplete third-order tensors in the generic case. We remark that even a modest improvement on the value of $r$ has potential consequences for downstream applications in learning, since it often imposes an upper bound on the complexity of models we can work with (e.g., the number of components in a mixture model). And the assumption that the components are random often leads to much more stringent assumptions on the model as well, at least compared to the generic case. 

For some time, the question of finding any algorithm (or an efficient uniqueness theorem) in the ``overcomplete'' regime $r > n$ remained a notable open problem (highlighted at COLT 2014~\cite{colt-open-prob} and also mentioned in~\cite{MSS,moitra-book}), and recently some progress has been made. The work of~\cite{CR-overcomplete} gave an algorithm for decomposing slightly overcomplete tensors, namely $r = n+\ell$ where the runtime is polynomial in $n$ but exponential in $\ell$. More recently, Koiran~\cite{koiran} gave both a decomposition algorithm and an efficient uniqueness theorem for $r \le 4n/3$; this algorithm only uses 4 ``slices'' of the tensor, and works for tensors of format $n \times n \times p$ for $p \ge 4$.
Our work aims to push these approaches to the limit, improving the condition on $r$ as much as possible and also investigating whether there are fundamental limitations for efficient algorithms.

\subsection{Our Contributions: Upper and Lower Bounds.}

Our starting point is the work of Garg, Kayal and Saha~\cite{garg2020learning} who showed a surprising connection between algebraic complexity theory and learning. In particular they showed how a powerful lower bound method, called the method of \emph{shifted partial derivatives}~\cite{chen2011partial}, could be leveraged to give algorithms for decomposing polynomials into power sums, which is a close relative of tensor decomposition. We will similarly build our algorithms out of certain algebraic gadgets, which we formalize below:

\begin{definition}\label{def:rankgadget}
A rank detection gadget (a.k.a.\ ``flattening'') $M$ is an efficiently computable function that maps third-order tensors $T \in \RR^{n \times n \times n}$ to matrices $M(T)$. We say that $M$ solves the generic rank detection problem up to rank $r$ if there is some function $f$ such that:
\begin{enumerate}
    \item For all third-order tensors $T$, $f(\rank(M(T)) \leq \rank(T)$.
    \item Moreover for generic tensors $T$ up to rank $r$, $f(\rank(M(T)) = \rank(T)$.
\end{enumerate}
\end{definition}
\noindent This definition is inspired by the definition of what it means to refute a random CSP. In particular we want that $f(\rank(M(T))$ should always yield a valid lower bound on the rank of $T$, but we only require this lower bound to be accurate on generic tensors up to some given rank. The main question is: How can we design rank detection gadgets that work for the largest range of ranks? 

A simple rank detection gadget is to just take the first ``slice" through the tensor. It is easy to check that its rank is always a lower bound on the rank of $T$. But since $M(T)$ is an $n \times n$ matrix, this gadget can only plausibly work for generic tensors up to rank $n$. Devising rank detection gadgets that work beyond $n$, i.e.\ in the overcomplete regime, turns out to be rather non-trivial. The first rank detection gadget that plausibly works beyond $n$ was given by Strassen \cite{strassen1983rank}. The key property is that if $T$ is $n \times n \times n$ then $M(T)$ is a $3n \times 3n$ matrix and moreover the mapping is linear and sends a rank-one tensor $T$ to a rank-two matrix $M(T)$. Thus Strassen's gadget plausibly works up to rank $r \leq 3n/2$. But proving this requires showing that the rank two matrices that arise from each component do not generically interfere with each other. Strassen's gadget was analyzed in the PhD thesis of Persu~\cite{persu-thesis} who showed it solves rank detection (but not decomposition) when $r \le 3n/2$. It turns out this is an example of an even more general strategy called \emph{Koszul--Young flattenings}~\cite{LO-young} (see also~\cite{landsberg2011tensors}), which are the best known rank detection gadgets. These are parameterized by a positive integer $p$ and map an $n \times n \times n$ tensor $T$ to a $\binom{2p+1}{p}n \times \binom{2p+1}{p}n$ matrix $M(T)$. They are again linear, and map a rank-one tensor $T$ to a rank-$\binom{2p}{p}$ matrix $M(T)$, so they plausibly work up to $r = (2-\epsilon) n$ for any $\epsilon > 0$. 

Our main algorithmic results are tensor decomposition and efficient uniqueness theorems that reach the same limit as the Koszul--Young flattenings. Below we give an informal and simplified statement of the main algorithm results (``upper bounds''), with the full details presented in Section~\ref{sec:uniqueness-thm}.
\begin{theorem*}[Upper Bounds, Informal]
Consider tensors of format $n_1 \times n_2 \times n_3$ where (without loss of generality) $n_1 \le n_2 \le n_3$, and further assume an asymptotic regime where $n_1 \to \infty$ and $n_3/n_2 =: \alpha = O(1)$. In the setting of generic components, we give algorithms for both the rank detection and decomposition tasks, as well as an efficient uniqueness theorem. These results tolerate rank as large as $r \le (1-\epsilon)(n_2 + n_3)$ for an arbitrary constant $\epsilon > 0$. For any fixed $\epsilon$, the runtime is polynomial in $n_3$. More precisely, the runtime is polynomial in $n_3 \cdot 2^{\alpha/\epsilon}$.
\end{theorem*}
\noindent For comparison, the previous state-of-the-art results for rank detection~\cite{persu-thesis} and decomposition~\cite{koiran} assume $n_2 = n_3 = n$ and require $r \le 3n/2$ and $r \le 4n/3$ respectively. Also, our efficient uniqueness theorem is the first uniqueness theorem (efficient or not) to surpass Kruskal's classical bound $r \le 3n/2$~\cite{kruskal}.

In addition to the Koszul--Young flattenings, a key ingredient in our approach will be the algorithm of~\cite{JLV} (see also the precursors~\cite{cardoso,DL-link,foobi,cpd-constrained}) for finding rank-1 matrices in a linear subspace. In contrast to our approach based on rank detection gadgets, the decomposition algorithm of~\cite{koiran} takes a different approach, based on connections with the problem of computing a \emph{commuting extension}. We can exploit this connection in the other direction: Appendix~\ref{app:comm-ext} shows how our decomposition algorithm implies an algorithm for commuting extensions that is (in some sense) nearly optimal.

We also investigate limitations of our approach (``lower bounds'') by showing that certain natural strategies cannot even solve the easier rank detection problem. We give a few different results for successively broader classes of flattenings. First we consider linear flattenings of the form $M(a \otimes b \otimes c) = A(a) \otimes (bc^\top)$ for a matrix $A(a)$ whose entries are homogeneous degree-1 polynomials in the entries of $a$. We designate such flattenings as \emph{Koszul--Young type} because the Koszul--Young flattenings take this form. We also consider the substantially more general \emph{degree-$d$ polynomial flattenings} where each entry of $M(T)$ is a homogeneous degree-$d$ polynomial in the entries of $T$. In both cases, we do not quite rule out rank detection in the sense of Definition~\ref{def:rankgadget}, but we rule out a natural sufficient condition called \emph{rank additivity}. In the case of linear flattenings, rank additivity corresponds to Definition~\ref{def:rankgadget} with $f(x) = mx$ for some $m \in \ZZ_{\ge 1}$. Below we give an informal and simplified statement of our lower bounds, with the full details presented in Section~\ref{sec:lower}.
\begin{theorem*}[Lower Bounds, Informal]
Consider tensors of format $n_1 \times n_2 \times n_3$ where $n_1 \le n_2 \le n_3$.
\begin{itemize}
    \item Flattenings of Koszul--Young type can only exhibit rank additivity when $r \le n_2 + n_3$.
    \item Degree-$d$ homogeneous polynomial flattenings can only exhibit rank additivity when $r \le Cn_3$ for a constant $C = C(d)$.
\end{itemize}
\end{theorem*}
\noindent The bound $n_2+n_3$ in the first result is precisely the threshold achieved by our algorithm, so the Koszul--Young flattenings are optimal among flattenings of this style. The proofs extend~\cite{barriers-rank}, which covers the case $d=1$ of the second result. We note that this style of lower bound is distinct from the low-degree polynomial framework used by~\cite{tensor-decomp-LD} for the case of random components, as that approach only applies to fully-Bayesian random models.

Our lower bounds formalize the sense in which reaching $r \gg n$ would require a rather different approach, and might be fundamentally out of reach for efficient algorithms. Of course, there are other potential approaches that are not ruled out by our lower bounds. For instance, we have not ruled out a flattening whose rank tracks that of the original tensor in a non-additive way. Beyond flattenings, there is the commuting extension approach of~\cite{koiran}, but curiously this appears to have a natural breaking point at $r=2n$ (see Section~2 of~\cite{koiran-extensions}), matching our algorithm. We note that there is an explicit tensor with a rank lower bound of roughly $3n$~\cite{factor-3}, but this does not appear to suggest a decomposition algorithm.

\paragraph{Model of computation.}

Our algorithm only uses standard linear algebra primitives such as solving linear equations, computing eigenvalues, and intersecting linear subspaces (all applied to polynomial-sized inputs). The eigenvalue computation is used within simultaneous diagonalization, which is a subroutine of~\cite{JLV}. Following~\cite[Section~1.4]{koiran} (and numerous prior works in the tensor decomposition literature), we treat these as atomic operations and assume they can be implemented exactly. Our algorithm can be implemented to give guarantees in the computational model that allows rational arithmetic with polynomial bits of precision.

We leave for future work the question of obtaining a detailed analysis of the stability and robustness of our algorithm, e.g., in the style of~\cite{smoothed,goyal-robust}. Certainly our description and analysis of the algorithm in exact terms is a necessary step toward this goal. Another seemingly necessary step would be to obtain a robust analysis of the~\cite{JLV} algorithm (which we use as a subroutine), but currently this only exists for the ``non-planted'' variant of the problem~\cite{robust-JLV} (where the goal is to \emph{refute} the existence of a rank-1 matrix in the given subspace).

\subsection{Further Discussion}
\label{sec:further}

We remark that there are fourth-order tensor decomposition algorithms that work in the generic and highly  overcomplete setting. In particular the \emph{FOOBI} algorithm and its variants~\cite{cardoso,DL-link,foobi,MSS,HSS-robust,JLV} can decompose generic tensors of rank $r \le cn^2$ for a constant $c > 0$. See~\cite{JLV} for a recent account, which generalizes the setting and corrects an error in the original analysis of~\cite{DL-link}; this result also certifies uniqueness of the decomposition, providing an efficiently-checkable uniqueness theorem. However from the vantage point of Definition~\ref{def:rankgadget} an associated rank detection gadget for fourth-order tensors is easy to construct: We can rearrange the $n \times n \times n \times n$ tensor into an $n^2 \times n^2$ matrix. This works as a rank detecting gadget for rank $r \leq n^2$. Also noise-robustness and numerical precision of (variants of) the above algorithms have been studied (for instance, in the smoothed analysis model) for order-3~\cite{smoothed,goyal-robust,koiran-precision} and order-4~\cite{MSS,HSS-robust} tensors.

\section{Main Results}

\subsection{Preliminaries}

\subsubsection{Tensor Decomposition and Rank}

\begin{definition}[Tensor Rank]
For vectors $a^{(1)} \in \RR^{n_1}, \, a^{(2)} \in \RR^{n_2}, \ldots, \, a^{(k)} \in \RR^{n_k}$, the \emph{rank-1} tensor $a^{(1)} \otimes a^{(2)} \otimes \cdots \otimes a^{(k)}$ has order $k$ and format $n_1 \times \cdots \times n_k$, with entries $(a^{(1)} \otimes \cdots \otimes a^{(k)})_{i_1,\ldots,i_k} = a^{(1)}_{i_1} a^{(2)}_{i_2} \cdots a^{(k)}_{i_k}$. In general, the \emph{rank} of an $n_1 \times \cdots \times n_k$ tensor $T$ is the minimum $r$ for which $T$ can be expressed as the sum of $r$ rank-1 tensors: $T = \sum_{\ell=1}^r a^{(1,\ell)} \otimes \cdots \otimes a^{(k,\ell)}$ where $a^{(i,\ell)} \in \RR^{n_i}$.
\end{definition}

\begin{definition}[Unique Decomposition]
\label{def:unique}
For an $n_1 \times \cdots \times n_k$ tensor $T$ and a positive integer $r$, we say that a decomposition $T = \sum_{\ell=1}^r a^{(1,\ell)} \otimes \cdots \otimes a^{(k,\ell)}$ into $r$ rank-1 terms is \emph{unique} if for any decomposition $T = \sum_{\ell=1}^r \tilde{a}^{(1,\ell)} \otimes \cdots \otimes \tilde{a}^{(k,\ell)}$ there exists a permutation $\pi: [r] \to [r]$ such that $\tilde{a}^{(1,\ell)} \otimes \cdots \otimes \tilde{a}^{(k,\ell)} = a^{(1,\pi(\ell))} \otimes \cdots \otimes a^{(k,\pi(\ell))}$ for all $\ell \in [r]$. When we speak of \emph{recovering} such a decomposition, we mean the task of finding the $r$ rank-1 tensors (which can be listed in any order), given $T$ as input.
\end{definition}

\noindent Note that if $T$ has a unique decomposition into $r$ rank-1 terms then none of these terms can be zero, and the rank of $T$ must be exactly $r$.

\subsubsection{Genericity and Symbolic Rank}
\label{sec:genericity}

Throughout, we use the standard meaning of the term ``generic'' from algebraic geometry. Namely, for symbolic variables $x = (x_1,\ldots,x_p)$, we say that a property holds true ``generically'' (or holds for ``generically chosen $x$'') if there exists a nonidentically-zero polynomial $P(x_1,\ldots,x_p) \in \RR[x]$ such that the property holds true for all $x \in \RR^n$ such that $P(x) \ne 0$. If a property holds true generically, it holds for ``almost all'' $x$ in the sense that the set of $x$ for which the property fails is measure zero. If a finite collection of properties each holds true generically, then it is also generically true that all the properties hold simultaneously; this can be seen by multiplying together all the associated polynomials $P(x)$.

For us, a common use case for genericity will be in reference to a tensor decomposition $T = \sum_{\ell=1}^r a^{(1,\ell)} \otimes \cdots \otimes a^{(k,\ell)}$ where $a^{(i,\ell)} \in \RR^{n_i}$. When we say a property holds true for ``generically chosen components,'' we mean the notion of generic from above, where the variables $x$ are all the entries of all the vectors $a^{(i,\ell)}$ (so $r \sum_{i=1}^k n_i$ variables in total).

Suppose $M = M(x)$ is a matrix whose entries are polynomials in the variables $x = (x_1,\ldots,x_p)$, that is, $M \in \RR[x]^{m \times n}$. Every minor of $M$ (determinant of a square submatrix) is a polynomial in the variables $x$. The \emph{symbolic rank} of $M$ is the largest integer $r$ for which $M$ has an $r \times r$ minor that is not identically zero. Equivalently, this is the rank of $M$ over the fraction field $\RR(x)$. Note that if $M$ has symbolic rank $r$ then $M(x)$ has rank $r$ for generically chosen $x$; this follows from the characterization of rank as the size of the largest nonzero minor. This allows us to speak of the ``generic rank'' of a matrix (which equals the symbolic rank). For specific values of $x$, the rank of $M(x)$ may be lower than the generic rank of $M$, but never higher (since all minors larger than the generic rank are identically zero).

\subsection{Upper Bounds}

\subsubsection{Trivial Flattenings}
\label{sec:trivial}

As a benchmark, we first discuss the ``trivial'' flattenings of tensors to matrices. Suppose $T$ is a tensor of format $n_1 \times \cdots \times n_k$. Choose a subset $S \subseteq [k]$ and let $\bar{S} := [k] \setminus S$ denote its complement. We can flatten $T$ to a matrix $M^\triv = M^\triv(T;S)$ of format $(\prod_{i \in S} n_i) \times (\prod_{i \in \bar{S}} n_i)$ in a straightforward way:
\begin{equation}\label{eq:def-triv}
M^\triv_{(i_j \,:\, j \in S), \, (i_j \,:\, j \in \bar{S})} := T_{i_1,\ldots,i_k} \qquad \text{where } i_j \in [n_j].
\end{equation}
A rank-1 tensor $a^{(1)} \otimes \cdots \otimes a^{(k)}$ flattens to a rank-1 matrix
\[ M^\triv(a^{(1)} \otimes \cdots \otimes a^{(k)};S) = \left(\bigotimes_{i \in S} a^{(i)}\right)\left(\bigotimes_{i \in \bar{S}} a^{(i)}\right)^\top. \]
A rank-$r$ tensor $T = \sum_{\ell=1}^r a^{(1,\ell)} \otimes \cdots \otimes a^{(k,\ell)}$ flattens to the matrix
\[ \sum_{\ell=1}^r M^\triv(a^{(1,\ell)} \otimes \cdots \otimes a^{(k,\ell)};S), \]
which has rank at most $r$. For rank detection, it would be useful if the rank were exactly $r$, i.e., if the rank of the flattening were \emph{additive} across the rank-1 terms of the tensor, so that we could read off the rank of $T$ from the rank of $M^\triv$. The rank of $M^\triv(T;S)$ certainly cannot exceed the minimum of its two dimensions, namely
\[ R_S := \min\left\{\prod_{i \in S} n_i, \, \prod_{i \in \bar{S}} n_i\right\}. \]
Conceivably, additivity could hold as long as $r \le R_S$. Indeed, this is the case as long as the components $a^{(i,\ell)}$ are generically chosen in the sense of Section~\ref{sec:genericity}; we include the proof in Appendix~\ref{app:trivial} for completeness. This gives a solution to the rank detection problem for tensors of rank up to $R_S$ (with generic components). Taking the best choice of $S$, the trivial flattening can solve rank detection up to rank
\[ n_* := \max_{S \subseteq [k]} \, \min\left\{\prod_{i \in S} n_i, \, \prod_{i \in \bar{S}} n_i\right\}. \]
For order-$k$ tensors of format $n \times \cdots \times n$, this gives rank detection up to rank $n^{\lfloor k/2 \rfloor}$.

Section~6 of~\cite{JLV} uses the trivial flattening to give an algorithm for recovering the minimum-rank decomposition (and certifying its uniqueness), under slightly stronger assumptions. The idea is to search the column span of the flattening for vectors that have rank-1 structure. For instance, Corollary~27 of~\cite{JLV} can decompose an $n_1 \times n_2 \times n_3$ tensor with $n_1 \le n_2 \le n_3$ and generically chosen components, of rank $r$ up to
\begin{equation}\label{eq:r-triv}
r \le \min\left\{\frac{1}{4}(n_1-1)(n_2-1), \, n_3\right\}.
\end{equation}
For comparison, $n_*$ takes the value $\min\{n_1 n_2,\, n_3\}$ in this setting.

\subsubsection{Koszul--Young Flattenings}

Our new algorithm is inspired by an idea first introduced by~\cite{LO-young} as ``Young flattenings''~\cite{LO-young,landsberg2011tensors,nontriv-explicit} and later referred to as ``Koszul flattenings''~\cite{alg-waring,border-matrix-mult,homotopy-decomp} or ``Koszul--Young flattenings''~\cite{KY-border,KY-complexity}. (We thank J.M.\ Landsberg for clarifying that Koszul flattenings are a special case of Young flattenings, with the optimal choice of parameters.) Essentially, Koszul--Young flattenings provide a non-trivial linear map from tensors to matrices. These ideas have already appeared in tensor decomposition algorithms~\cite{alg-waring}, but for $n \times n \times n$ tensors, these results only handle the undercomplete case $r \le n$. More recently,~\cite{persu-thesis} used Koszul--Young flattenings to solve rank detection for overcomplete tensors with generic components, namely $n \times n \times n$ tensors of rank up to $3n/2$. Furthermore,~\cite{persu-thesis} proposed a hierarchy of poly-time approaches for rank detection that conjecturally approaches rank $2n$, as well as an approach for decomposition that conjecturally tolerates rank $(1+\epsilon)n$ for a small constant $\epsilon > 0$. We will resolve and extend these conjectures: for sufficiently large $n \times n \times n$ tensors with generic components, we solve both rank detection and decomposition in polynomial time, up to rank $(2-\epsilon)n$ for an arbitrary constant $\epsilon > 0$. We will also handle non-square tensors.

The Koszul--Young flattenings are typically described in abstract algebraic terms, namely as a map $B^* \otimes \wedge^{p}(A) \rightarrow \wedge^{p+1}A \otimes C$. We will give instead a very explicit combinatorial description of a matrix based on these ideas. Our construction has two parameters $p,q$ instead of only $p$, and specializes to (the transpose of) the usual Koszul--Young flattenings when $q = 2p+1$. Suppose $T$ is $n_1 \times n_2 \times n_3$ and fix integers $p,q$ with $0 \le p < p+1 \le q \le n_1$. Define a matrix $M = M(T;p,q)$ with rows indexed by
\[ \{(S,j) \,:\, S \subseteq [q], \, |S|=p, \, j \in [n_2]\}, \]
columns indexed by
\[ \{(U,k) \,:\, U \subseteq [q], \, |U|=p+1, \, k \in [n_3]\}, \]
and entries
\begin{equation}\label{eq:def-M}
M_{Sj,Uk} := \sum_{i=1}^q \One_{U = S \sqcup \{i\}} \cdot \sigma(U,i) \cdot T_{ijk}.
\end{equation}
Here, $\One_A$ denotes the $\{0,1\}$-valued indicator for an event $A$, and $\sigma(U,i) \in \{\pm 1\}$ is the parity of $i$'s position in $U$, that is,
\begin{equation}\label{eq:def-sig}
\sigma(U,i) := (-1)^{|\{j \in U \,:\, j < i\}|}.
\end{equation}
While we have written~\eqref{eq:def-M} as a sum over $i \in [q]$, at most one term in this sum can be nonzero. Note that only $q$ slices of $T$ (along the first mode) are used.

The dimensions of $M$ are $\binom{q}{p} n_2 \times \binom{q}{p+1} n_3$. We will leave $p,q$ as free variables for now, but eventually it will be advantageous to take $q$ to be a (large) constant so that $M$ has polynomial size, and $p \approx q \frac{n_3}{n_2+n_3}$ which makes $M$ as square as possible. One special case of interest is $n_2 = n_3$, in which case we can take $q$ odd with $q = 2p+1$ so that $M$ is exactly square, since $\binom{2p+1}{p} = \binom{2p+1}{p+1}$; the prior work~\cite{persu-thesis} considers only this square case and mainly focuses on $p=1$. We remark that the extreme cases $p=0$ and $p=q-1$ essentially reduce to ``trivial'' flattenings from the previous section (but with additional signs in the case $p=q-1$).

The flattening~\eqref{eq:def-M} bears resemblance to a signed variant of the \emph{Kikuchi matrices} that were introduced by~\cite{kikuchi} (with a variation independently discovered by~\cite{hastings-quantum}) and have subsequently appeared in results on constraint satisfaction~\cite{smoothed-random}, error-correcting codes~\cite{ldc,lcc-1,lcc-2}, and extremal combinatorics~\cite{smoothed-random,simpler-moore}. In particular, the related matrix $A$ defined in~\eqref{eq:def-A} below resembles a Kikuchi matrix in the sense that its rows and columns are indexed by subsets and it only contains a nonzero entry when the two subsets have their symmetric difference of a particular size (in this case, 1).

\subsubsection{Rank Detection}
\label{sec:rank-det}

If $T$ is a rank-1 tensor $T = a \otimes b \otimes c$ then the flattening $M = M(T;p,q)$ from~\eqref{eq:def-M} is the Kronecker product $M = A \otimes (bc^\top)$ where $A = A(a;p,q)$ is the $\binom{q}{p} \times \binom{q}{p+1}$ matrix
\begin{equation}\label{eq:def-A}
A_{SU} = \sum_{i=1}^q \One_{U = S \sqcup \{i\}} \cdot \sigma(U,i) \cdot a_i.
\end{equation}
Therefore, as long as neither $b$ nor $c$ is the zero vector, $M$ has the same rank as $A$. We will see that the matrix $A$ is rank-deficient, with generic rank $\binom{q-1}{p}$. If, for instance, $q = 2p+1$ then $A$ is square since $\binom{2p+1}{p} = \binom{2p+1}{p+1}$, but crucially, its generic rank $\binom{2p}{p}$ is roughly half its dimension (for large $p$). This is a known fact, and we provide a proof in Section~\ref{sec:pf-prop-rank}.

\begin{proposition}[\cite{LO-young}]
\label{prop:rank}
For a rank-1 tensor $T = a \otimes b \otimes c$ with $a,b,c$ generically chosen, the matrix $M(T;p,q)$ defined in~\eqref{eq:def-M} has rank exactly $\binom{q-1}{p}$.
\end{proposition}

\noindent The transformation $T \mapsto M(T;p,q)$ is linear, so a rank-$r$ tensor $T = \sum_{\ell=1}^r a^{(\ell)} \otimes b^{(\ell)} \otimes c^{(\ell)}$ flattens to the matrix
\[ M(T;p,q) = \sum_{\ell=1}^r M(a^{(\ell)} \otimes b^{(\ell)} \otimes c^{(\ell)};p,q). \]
If the rank were additive across these $r$ terms then $M(T;p,q)$ would need to have rank $r\binom{q-1}{p}$, in light of Proposition~\ref{prop:rank}. Conceivably, additivity could hold as long as $r\binom{q-1}{p}$ does not exceed the minimum dimension of $M(T;p,q)$, which is $\min\{\binom{q}{p}n_2,\binom{q}{p+1}n_3\}$. Dividing through by $\binom{q-1}{p}$, this condition becomes $r \le \min\{\frac{q}{q-p} n_2, \frac{q}{p+1} n_3\}$. Choosing $q$ large and $p \approx q \frac{n_3}{n_2+n_3}$ to balance the two terms in $\min\{\cdots\}$, this heuristic suggests that rank detection might be possible for $r$ as large as, roughly, $n_2+n_3$. We show that indeed this pans out.

\begin{theorem}[Rank detection]
\label{thm:rank-det}
Let $1 \le q \le n_1$ and
\[ p = \left\lfloor q \cdot \frac{n_3}{n_2+n_3} \right\rfloor. \]
If $T = \sum_{\ell=1}^r a^{(\ell)} \otimes b^{(\ell)} \otimes c^{(\ell)}$ is $n_1 \times n_2 \times n_3$ with generically chosen components, and
\[ r \le (n_2+n_3)\left(1 - \frac{1+\alpha}{q}\right) - q \qquad \text{where } \alpha := \max\left\{\frac{n_2}{n_3},\frac{n_3}{n_2}\right\}, \]
then the matrix $M(T;p,q)$ defined in~\eqref{eq:def-M} has rank exactly $r \binom{q-1}{p}$.
\end{theorem}

\noindent The proof can be found in Section~\ref{sec:pf-rank-det}. Of course, one has the option to permute the 3 modes of the tensor before applying Theorem~\ref{thm:rank-det}, and it is advantageous to make $n_1$ the smallest mode. Suppose $n_1 \le n_2 \le n_3$ and consider the regime $n_1 \to \infty$ with $\alpha$ of constant order. For any constant $\epsilon > 0$ we can handle $r$ as large as $(1-\epsilon)(n_2+n_3)$ by taking $q$ to be a sufficiently large constant (depending on $\alpha,\epsilon$). For comparison, the trivial flattening of Section~\ref{sec:trivial} works up to rank $n_3$ in this regime. Our advantage over the trivial flattening diminishes as the ratio $\alpha = n_3/n_2$ increases. As long as $q$ is a constant, the size of $M$ and the complexity of computing its entries are polynomial in the largest dimension $n_3$.

How does the runtime scale with $\epsilon$? Since $q$ is a constant, it suffices to have $(1+\alpha)/q \le \epsilon/2$, meaning we should choose $q \ge 2(1+\alpha)/\epsilon$. Now each dimension of $M$ is bounded by $n_3 \cdot 2^q = n_3 \cdot 2^{O(\alpha/\epsilon)}$. The runtime is therefore polynomial in $n_3 \cdot 2^{\alpha/\epsilon}$. It might be possible to remove the dependence on $\alpha$ with a more careful analysis.

\begin{remark}[Certifying lower bounds on tensor rank]
For an \emph{arbitrary} tensor $T$, we always have $\rank(M(T;p,q)) \le \binom{q-1}{p} \cdot \rank(T)$. This follows from the fact that the rank of $A$ from~\eqref{eq:def-A} \emph{never} exceeds its generic rank $\binom{q-1}{p}$; see Lemma~\ref{lem:rank-A}. As a result, $M$ provides a \emph{certificate} that $\rank(T) \ge \lceil \rank(M(T;p,q)) / \binom{q-1}{p} \rceil$ for any given tensor $T$.
\end{remark}

\subsubsection{Decomposition Algorithm}

We now outline our strategy for how to use the flattening~\eqref{eq:def-M} to extract the components of the tensor. The aim of this section is to describe the algorithm and the intuition behind it, while Section~\ref{sec:uniqueness-thm} below will contain the main theorems that guarantee this algorithm's success when the tensor components are generically chosen.

Consider an $n_1 \times n_2 \times n_3$ tensor $T = \sum_{\ell=1}^r a^{(\ell)} \otimes b^{(\ell)} \otimes c^{(\ell)}$. Build the matrix $M = M(T;p,q)$ defined in~\eqref{eq:def-M} and compute its column span. Recall that each rank-1 term $a^{(\ell)} \otimes b^{(\ell)} \otimes c^{(\ell)}$ flattens to $A(a^{(\ell)};p,q) \otimes (b^{(\ell)}c^{(\ell)\top})$, where $A(a;p,q)$ is defined in~\eqref{eq:def-A}. Generically, we have the additivity property of Theorem~\ref{thm:rank-det}, implying that the column span of $M$ is equal to the span of all columns of these individual flattenings, that is,
\begin{equation}\label{eq:col-span}
\Colspan(M) = \Span\{z^{(U)}(a^{(\ell)};p,q) \otimes b^{(\ell)} \,:\, U \subseteq [q], \, |U| = p+1, \, \ell \in [r]\}
\end{equation}
where $z^{(U)}(a;p,q) \in \RR^{\binom{q}{p}}$ denotes column $U$ of $A(a;p,q)$.

Fix one particular column of $A$, namely $V := [p+1]$. The column $z^{(V)}$ has a particular sparsity pattern, with nonzeros only allowed in the $p+1$ rows $S$ where $S \subseteq V$. Let $Z_{p,q} \subseteq \RR^{\binom{q}{p}}$ denote the $(p+1)$-dimensional subspace comprised of all vectors with this sparsity pattern, that is,
\[ Z_{p,q} := \Span\{e^{(S)} \,:\, S \subseteq [q], \, |S| = p, \, S \subseteq [p+1]\} \subseteq \RR^{\binom{q}{p}}, \]
where $e^{(S)}$ denotes the standard unit basis vector. We will compute the intersection of the two subspaces $\Colspan(M)$ and $Z_{p,q} \otimes \RR^{n_2}$. By design, the $r$ vectors $z^{(V)}(a^{(\ell)};p,q) \otimes b^{(\ell)}$ for $\ell \in [r]$ lie in the intersection, and we will show that, generically, the intersection does not contain any additional ``spurious'' elements, that is,
\begin{equation}\label{eq:intersect-colspan}
\Colspan(M) \cap (Z_{p,q} \otimes \RR^{n_2}) = \Span\{z^{(V)}(a^{(\ell)};p,q) \otimes b^{(\ell)} \,:\, \ell \in [r]\}.
\end{equation}
It will be convenient to change basis so as to remove the zero entries of $z^{(V)}$. To this end, we will define a specific isomorphism $\varphi_{p,q}: Z_{p,q} \to \RR^{p+1}$, namely the linear map defined on unit basis vectors as follows:
\begin{equation}\label{eq:def-phi}
\varphi_{p,q}(e^{([p+1] \setminus \{i\})}) = (-1)^{i-1} e^{(i)} \qquad \text{for } i \in [p+1].
\end{equation}
This is designed so that $\varphi_{p,q}(z^{(V)}(a;p,q)) = (a_1,\ldots,a_{p+1})^\top$, since entry $V \setminus \{i\}$ of $z^{(V)}(a;p,q)$ is equal to $\sigma(V,i) a_i = (-1)^{i-1} a_i$. Defining $d^{(\ell)}$ as the first $p+1$ entries of $a^{(\ell)}$,
\[ d^{(\ell)} := (a^{(\ell)}_1,\ldots,a^{(\ell)}_{p+1})^\top, \]
this allows us to write~\eqref{eq:intersect-colspan} as
\begin{equation}\label{eq:apply-phi}
(\varphi_{p,q} \otimes I)(\Colspan(M) \cap (Z_{p,q} \otimes \RR^{n_2})) = \Span\{d^{(\ell)} \otimes b^{(\ell)} \,:\, \ell \in [r]\},
\end{equation}
where $I$ is the identity map $\RR^{n_2} \to \RR^{n_2}$.

To recap, we can compute the left-hand side of~\eqref{eq:apply-phi} from the flattening $M$, which gives us access to the right-hand side. This means we know the span of $r$ different rank-1 tensors $d^{(\ell)} \otimes b^{(\ell)}$ (which can also be viewed as rank-1 matrices of dimension $(p+1) \times n_2$) and our next objective will be to extract the individual rank-1 elements $d^{(\ell)} \otimes b^{(\ell)}$. Since our rank-1 tensors are generic, the work of~\cite{JLV} provides an algorithm for precisely this task, under the condition $r \le p(n_2-1)/4$; see Corollary~3 of~\cite{JLV}, which considers a generalization of this task (in their notation, we only need the case $s=R$). This allows us to recover vectors $\hat{d}^{(\ell)}$ and $\hat{b}^{(\ell)}$ for $\ell \in [r]$, such that $\hat{d}^{(\ell)} = \alpha_\ell \cdot d^{(\pi(\ell))}$ and $\hat{b}^{(\ell)} = \beta_\ell \cdot b^{(\pi(\ell))}$ where $\alpha_\ell, \beta_\ell \in \RR$ are nonzero scalars and $\pi$ is a permutation $[r] \to [r]$.

At this point, we repeat the entire algorithm so far with the ``$b$'' and ``$c$'' modes switched. We use the same value for $q$, but use $q-p-1$ in place of $p$ so that the new flattening remains as square as possible. This allows us to obtain vectors $\hat{c}^{(\ell)}$ and $\hat{f}^{(\ell)}$ which serve as estimates for $c^{(\ell)}$ and
\[ f^{(\ell)} := (a^{(\ell)}_1,\ldots,a^{(\ell)}_{q-p})^\top, \]
respectively (again, up to re-ordering and scalar multiple). By comparing the $\hat{d}^{(\ell)}$ and $\hat{f}^{(\ell)}$ vectors we can identify the correct pairing between $\hat{b}^{(\ell)}$ and $\hat{c}^{(\ell)}$, and re-index the $\hat{c}^{(\ell)}$ so that $\hat{b}^{(1)}$ is paired with $\hat{c}^{(1)}$, etc. Now we know the decomposition takes the form
\[ T = \sum_{\ell=1}^r w^{(\ell)} \otimes \hat{b}^{(\ell)} \otimes \hat{c}^{(\ell)} \]
for unknown vectors $w^{(\ell)}$, which can be recovered by solving a linear system of equations; we will show that, generically, this linear system has a unique solution, namely the desired decomposition.

The full algorithm is presented as Algorithm~\ref{alg:decomp} below. The second phase of the algorithm involves a flattening $M' = M'(T;p,q)$ that is identical to $M$ but with the ``$b$'' and ``$c$'' modes switched, and with $q-p-1$ in place of $p$. Formally, $M'(T;p,q)$ is indexed by
\[ \{(S,k) \,:\, S \subseteq [q], \, |S| = q-p-1, \, k \in [n_3]\} \times \{(U,j) \,:\, U \subseteq [q], \, |U| = q-p, \, j \in [n_2]\} \]
with entries
\begin{equation}\label{eq:def-M-prime}
M'_{Sk,Uj} := \sum_{i=1}^q \One_{U = S \sqcup \{i\}} \cdot \sigma(U,i) \cdot T_{ijk}.
\end{equation}

\begin{algorithm}[Overcomplete third-order tensor decomposition]
\label{alg:decomp}
\phantom{}
\begin{itemize}
\item Input: A tensor $T \in \RR^{n_1 \times n_2 \times n_3}$.
\item Input: Parameters $p,q \in \ZZ$ with $q \le n_1$ and $\bar{p} := \min\{p+1,q-p\} \ge 2$.
\item Input: A parameter $r$ (or learn $r$ from $\rank(M)$ using Theorem~\ref{thm:rank-det}).
\end{itemize}
\begin{enumerate}
\item Build the matrix $M = M(T;p,q)$ as in~\eqref{eq:def-M}.
\item Construct the subspace $B := (\varphi_{p,q} \otimes I)(\Colspan(M) \cap (Z_{p,q} \otimes \RR^{n_2})) \subseteq \RR^{p+1} \otimes \RR^{n_2}$.
\item Run the algorithm of~\cite[Corollary~3]{JLV} to find rank-1 tensors in $B$, call them $\hat{d}^{(1)} \otimes \hat{b}^{(1)}, \ldots, \hat{d}^{(r)} \otimes \hat{b}^{(r)}$. If the number of rank-1 tensors found is not $r$, output ``fail.''
\item Build the matrix $M' = M'(T;p,q)$ as in~\eqref{eq:def-M-prime}.
\item Construct the subspace $C := (\varphi_{q-p-1,q} \otimes I)(\Colspan(M') \cap (Z_{q-p-1,q} \otimes \RR^{n_3})) \subseteq \RR^{q-p} \otimes \RR^{n_3}$.
\item Run the algorithm of~\cite[Corollary~3]{JLV} to find rank-1 tensors in $C$, call them $\hat{f}^{(1)} \otimes \hat{c}^{(1)}, \ldots, \hat{f}^{(r)} \otimes \hat{c}^{(r)}$. If the number of rank-1 tensors found is not $r$, output ``fail.''
\item\label{item:p-bar} Define a permutation $\tau: [r] \to [r]$ so that the first $\bar{p}$ entries of $\hat{f}^{(\tau(\ell))}$ form a scalar multiple of the first $\bar{p}$ entries of $\hat{d}^{(\ell)}$ for each $\ell \in [r]$. If this is not possible, output ``fail.''
\item Solve the following linear system of equations in variables $w^{(\ell)} \in \RR^{n_1}$ for $\ell \in [r]$:
\begin{equation}\label{eq:linear}
\sum_{\ell=1}^r w^{(\ell)} \otimes \hat{b}^{(\ell)} \otimes \hat{c}^{(\tau(\ell))} = T,
\end{equation}
and output the resulting rank-$r$ decomposition of $T$. If there is no solution, output ``fail.''
\end{enumerate}
\end{algorithm}

We have assumed $\bar{p} \ge 2$, since otherwise Step~\ref{item:p-bar} will not have a unique solution. If $\bar{p} = 1$, this means either $p=0$ or $p=q-1$, which are the extreme cases where $M(T;p,q)$ essentially reduces to a trivial flattening. In these cases, one should use a simpler decomposition algorithm of~\cite{JLV} based on the trivial flattening, as we mentioned in Section~\ref{sec:trivial}; this algorithm works under the condition~\eqref{eq:r-triv}.

\subsubsection{Uniqueness Theorem and Algorithmic Guarantees}
\label{sec:uniqueness-thm}

We will see that the discussion above gives not only a decomposition algorithm that works for generic rank-$r$ tensors but also an explicit list of efficiently-checkable conditions on the components that guarantees uniqueness of the decomposition. We first present these conditions (Theorem~\ref{thm:uniqueness}) and then show that generic components satisfy these conditions when $r$ is small enough (Theorem~\ref{thm:generic-success}).

Some of the conditions below may appear rather opaque, so we first briefly motivate what role they play. It is natural for the matrix $M$ to appear (condition~\ref{item:M}), given its use as a rank detection device (Theorem~\ref{thm:rank-det}) and its role in the decomposition algorithm, but we will also define two additional matrices $N$ and $P$. The condition~\ref{item:N} on $N$ is used to ensure that~\eqref{eq:intersect-colspan} holds (no spurious elements in the subspace intersection), while the condition~\ref{item:P} on $P$ is used to ensure that the~\cite{JLV} algorithm succeeds. The matrices $M', N', P'$ are identical to $M,N,P$ but with the ``$b$'' and ``$c$'' modes of the tensor switched and with $p$ replaced by $q-p-1$.

\begin{definition}
\label{def:uniqueness-def}
Consider an $n_1 \times n_2 \times n_3$ tensor $T$ with decomposition
\begin{equation}
T = \sum_{\ell=1}^r a^{(\ell)} \otimes b^{(\ell)} \otimes c^{(\ell)}
\end{equation}
and parameters $p,q \in \ZZ$ with $q \le n_1$ and $\bar{p} := \min\{p+1,q-p\} \ge 2$.
\begin{itemize}
    \item The matrix $N$ has rows indexed by
    \[ \{(S,j) \,:\, S \subseteq [q], \, |S| = p, \, S \not\subseteq [p+1], \, j \in [n_2]\} \]
    and columns indexed by
    \[ \{(U,\ell) \,:\, U \subseteq [q], \, |U| = p+1, \, 1 \in U, \, U \ne [p+1], \, \ell \in [r]\}, \]
    with entries
    \begin{equation}\label{eq:def-N}
    N_{Sj,U\ell} = b^{(\ell)}_j \cdot \sum_{i=1}^q \One_{U = S \sqcup \{i\}} \cdot \sigma(U,i) \cdot a^{(\ell)}_i,
    \end{equation}
    where $\sigma$ is defined in~\eqref{eq:def-sig}.
    \item The matrix $N'$ has rows indexed by
    \[ \{(S,k) \,:\, S \subseteq [q], \, |S| = q-p-1, \, S \not\subseteq [q-p], \, k \in [n_3]\} \]
    and columns indexed by
    \[ \{(U,\ell) \,:\, U \subseteq [q], \, |U| = q-p, \, 1 \in U, \, U \ne [q-p], \, \ell \in [r]\}, \]
    with entries
    \[ N'_{Sk,U\ell} = c^{(\ell)}_k \cdot \sum_{i=1}^q \One_{U = S \sqcup \{i\}} \cdot \sigma(U,i) \cdot a^{(\ell)}_i. \]
    \item The matrix $P$ is indexed by
    \[ \{(i_1,i_2,j_1,j_2) \,:\, 1 \le i_1 < i_2 \le p+1, \, 1 \le j_1 < j_2 \le n_2\} \times \{(\ell_1,\ell_2) \,:\, 1 \le \ell_1 < \ell_2 \le r\} \]
    with entries
    \begin{equation}\label{eq:def-P}
    P_{i_1 i_2 j_1 j_2 , \ell_1 \ell_2} = a^{(\ell_1)}_{i_1} b^{(\ell_1)}_{j_1} a^{(\ell_2)}_{i_2} b^{(\ell_2)}_{j_2} + a^{(\ell_2)}_{i_1} b^{(\ell_2)}_{j_1} a^{(\ell_1)}_{i_2} b^{(\ell_1)}_{j_2} - a^{(\ell_1)}_{i_1} b^{(\ell_1)}_{j_2} a^{(\ell_2)}_{i_2} b^{(\ell_2)}_{j_1} - a^{(\ell_2)}_{i_1} b^{(\ell_2)}_{j_2} a^{(\ell_1)}_{i_2} b^{(\ell_1)}_{j_1}.
    \end{equation}
    \item The matrix $P'$ is indexed by
    \[ \{(i_1,i_2,j_1,j_2) \,:\, 1 \le i_1 < i_2 \le q-p, \, 1 \le j_1 < j_2 \le n_3\} \times \{(\ell_1,\ell_2) \,:\, 1 \le \ell_1 < \ell_2 \le r\} \]
    with entries
    \[ P'_{i_1 i_2 j_1 j_2 , \ell_1 \ell_2} = a^{(\ell_1)}_{i_1} c^{(\ell_1)}_{j_1} a^{(\ell_2)}_{i_2} c^{(\ell_2)}_{j_2} + a^{(\ell_2)}_{i_1} c^{(\ell_2)}_{j_1} a^{(\ell_1)}_{i_2} c^{(\ell_1)}_{j_2} - a^{(\ell_1)}_{i_1} c^{(\ell_1)}_{j_2} a^{(\ell_2)}_{i_2} c^{(\ell_2)}_{j_1} - a^{(\ell_2)}_{i_1} c^{(\ell_2)}_{j_2} a^{(\ell_1)}_{i_2} c^{(\ell_1)}_{j_1}. \]
\end{itemize}
\end{definition}

\begin{theorem}[Uniqueness theorem]
\label{thm:uniqueness}
Consider an $n_1 \times n_2 \times n_3$ tensor $T$ with decomposition
\begin{equation}\label{eq:rank-r-decomp}
T = \sum_{\ell=1}^r a^{(\ell)} \otimes b^{(\ell)} \otimes c^{(\ell)}
\end{equation}
and parameters $p,q \in \ZZ$ with $q \le n_1$ and $\bar{p} := \min\{p+1,q-p\} \ge 2$. Suppose the components satisfy the following conditions:
\begin{enumerate}[label=(\roman*)]
    \item\label{item:a1} $a^{(\ell)}_1 \ne 0$ for all $\ell \in [r]$.
    \item\label{item:d-distinct} The $r$ vectors $(a^{(\ell)}_1,\ldots,a^{(\ell)}_{\bar{p}})^\top$ for $\ell \in [r]$ are pairwise linearly independent, i.e., no one of these vectors is a scalar multiple of another.
    \item\label{item:db-li} The $r$ vectors $\{d^{(\ell)} \otimes b^{(\ell)} \,:\, \ell \in [r]\}$ are linearly independent, where $d^{(\ell)} := (a^{(\ell)}_1,\ldots,a^{(\ell)}_{p+1})^\top$.
    \item\label{item:dc-li} The $r$ vectors $\{f^{(\ell)} \otimes c^{(\ell)} \,:\, \ell \in [r]\}$ are linearly independent, where $f^{(\ell)} := (a^{(\ell)}_1,\ldots,a^{(\ell)}_{q-p})^\top$.
    \item\label{item:bc-li} The $r$ vectors $\{b^{(\ell)} \otimes c^{(\ell)} \,:\, \ell \in [r]\}$ are linearly independent.
    \item\label{item:M} The matrix $M = M(T;p,q)$ from~\eqref{eq:def-M} has rank exactly $r\binom{q-1}{p}$.
    \item\label{item:M-prime} The matrix $M' = M'(T;p,q)$ from~\eqref{eq:def-M-prime} has rank exactly $r\binom{q-1}{p}$.
    \item\label{item:N} The matrix $N$ from Definition~\ref{def:uniqueness-def} has full column rank.
    \item\label{item:N-prime} The matrix $N'$ from Definition~\ref{def:uniqueness-def} has full column rank.
    \item\label{item:P} The matrix $P$ from Definition~\ref{def:uniqueness-def} has full column rank.
    \item\label{item:P-prime} The matrix $P'$ from Definition~\ref{def:uniqueness-def} has full column rank.
\end{enumerate}
Then~\eqref{eq:rank-r-decomp} is the unique rank-$r$ decomposition of $T$ (in the sense of Definition~\ref{def:unique}), and this decomposition is recovered by Algorithm~\ref{alg:decomp} (with input $T,p,q,r$).
\end{theorem}

\noindent The proof of Theorem~\ref{thm:uniqueness} can be found in Section~\ref{sec:pf-uniqueness}. It may be possible to simplify the above conditions or to find conditions that are ``minimal'' in some sense, but we have not attempted this here. If the above conditions fail to hold for a particular decomposition, one can potentially remedy this by permuting the modes of the tensor or by applying an invertible change of basis on one or more of the modes. Next we show that the above conditions hold for ``almost all'' tensors of low enough rank.

\begin{theorem}[Success for generic tensors]
\label{thm:generic-success}
Consider the setting of Theorem~\ref{thm:uniqueness} with parameters $q \le n_1$ and
\[ p = \left\lfloor q \cdot \frac{n_3}{n_2+n_3} \right\rfloor. \]
Suppose
\[ q \ge (4+5\alpha)\left(1+\frac{1}{\alpha}\right) \qquad \text{where } \alpha := \max\left\{\frac{n_2}{n_3},\frac{n_3}{n_2}\right\}. \]
If the components $a^{(\ell)}, b^{(\ell)}, c^{(\ell)}$ are generically chosen with
\begin{equation}\label{eq:final-r-bound}
r \le (n_2+n_3)\left(1-\frac{3+\alpha}{q}\right) - \frac{q^3}{4}
\end{equation}
then $\bar{p} \ge 2$ and conditions~\ref{item:a1}--\ref{item:P-prime} are satisfied.
\end{theorem}

\noindent The proof can be found in Section~\ref{sec:pf-generic-success}. This is qualitatively the same behavior as our rank detection result. Suppose $n_1 \le n_2 \le n_3$ and consider the regime $n_1 \to \infty$ with $\alpha$ of constant order. For any constant $\epsilon > 0$ we can handle $r$ as large as $(1-\epsilon)(n_2+n_3)$ by taking $q$ to be a sufficiently large constant (depending on $\alpha,\epsilon$). As long as $q$ is a constant, the runtime of Algorithm~\ref{alg:decomp} is polynomial in the largest dimension $n_3$ (more precisely, polynomial in $n_3 \cdot 2^{\alpha/\epsilon}$).

\subsection{Lower Bounds}
\label{sec:lower}

We have given an algorithm for tensor decomposition based on flattening the tensor to a matrix in a non-trivial manner: $M = M(T)$. We now explore the inherent limitations of this style of approach. A seemingly necessary prerequisite to finding a better decomposition algorithm is to find a better algorithm for the easier task of rank detection. We therefore focus on rank detection, and specifically we will show that certain types of flattenings cannot have ``additivity of rank'' when the tensor rank is too large. This additivity property was key to our rank detection algorithm (Theorem~\ref{thm:rank-det}) and was also crucial for our decomposition algorithm because it ensures that the column span of $M(\sum_\ell a^{(\ell)} \otimes b^{(\ell)} \otimes c^{(\ell)})$ retains all vectors from the individual column spans of the terms $M(a^{(\ell)} \otimes b^{(\ell)} \otimes c^{(\ell)})$. We will give three different lower bounds that apply to increasingly general families of flattenings.

\subsubsection{Flattenings of Koszul--Young Type}

The flattenings used in our algorithm have the property that two modes are flattened in the trivial manner while the third is flattened according to a non-trivial linear map. Formally, we have a linear map $M: \RR^{n_1 \times n_2 \times n_3} \to \RR^{N_1 \times N_2}$ (or in other words, the entries of $M$ are degree-1 homogeneous polynomials in the variables $T_{ijk}$) with the property
\begin{equation}\label{eq:ky-type}
M(a \otimes b \otimes c) = A(a) \otimes (bc^\top)
\end{equation}
for a matrix $A(a)$ whose entries are degree-1 homogeneous polynomials in the entries of $a$. On the right-hand side above, the symbol $\otimes$ denotes the Kronecker product of matrices. We will first see that any flattening of this form cannot surpass rank $n_2 + n_3$, making our algorithm optimal up to a factor $1-\epsilon$ among such flattenings (in the regime $n_1 \le n_2 \le n_3$ with $n_1 \to \infty$ and $n_3/n_2 = O(1)$).

\begin{theorem}\label{thm:lower-ky}
Suppose $M: \RR^{n_1 \times n_2 \times n_3} \to \RR^{N_1 \times N_2}$ is a linear map with the property~\eqref{eq:ky-type}. Suppose further that for generically chosen components $a^{(\ell)}, b^{(\ell)}, c^{(\ell)}$ and some $r \ge 1$, we have rank-$r$ additivity, that is,
\begin{equation}\label{eq:def-additivity}
\rank\left(M\left(\sum_{\ell=1}^r a^{(\ell)} \otimes b^{(\ell)} \otimes c^{(\ell)}\right)\right) = r \cdot \rank\left(M\left(a^{(1)} \otimes b^{(1)} \otimes c^{(1)}\right)\right).
\end{equation}
Then
\[ r \le n_2 + n_3. \]
\end{theorem}

\noindent The proof can be found in Section~\ref{sec:pf-lower-ky}. We thank an anonymous reviewer for pointing out that the case $n_2=n_3$ of this result appears in a remark after Theorem~6.1 of~\cite{nc-rank}. Recall from Section~\ref{sec:genericity} that symbolic matrices have a well-defined notion of generic rank, so the expressions on the left- and right-hand sides of~\eqref{eq:def-additivity} must each take one specific value when the components are chosen generically. (We cannot have a situation where some positive-measure set of components leads to rank $r_1$ while some other positive-measure set of components leads to a different rank $r_2$.) The right-hand side of~\eqref{eq:def-additivity} could equivalently be written as
\[ \sum_{\ell=1}^r \rank\left(M\left(a^{(\ell)} \otimes b^{(\ell)} \otimes c^{(\ell)}\right)\right), \]
since for generic components, the $r$ terms in this sum are all equal.

\subsubsection{Linear Flattenings}

We now consider a more general class of flattenings, namely all linear flattenings (without the requirement~\eqref{eq:ky-type}). Our bound on the rank will be weaker than above.

\begin{theorem}\label{thm:lower-linear}
Suppose $M: \RR^{n_1 \times n_2 \times n_3} \to \RR^{N_1 \times N_2}$ is a linear map. Suppose further that for generically chosen components and some $r \ge 1$, we have rank-$r$ additivity as defined in~\eqref{eq:def-additivity}. Then
\[ r \le 2(n_1 + n_2 + n_3 + 1). \]
\end{theorem}

\noindent The proof can be found in Section~\ref{sec:pf-lower-linear}. This result is similar to Theorem~4.4 of~\cite{barriers-rank}, where the bound $r \le 8n$ is proven for ``square'' tensors, i.e., tensors of format $n \times n \times n$. For square tensors, our result gives an improvement from $8n$ to $6n+2$. Our argument is similar to that of~\cite{barriers-rank} and we describe the difference in Remark~\ref{rem:improve-8n}. We thank an anonymous reviewer for pointing out that the improvement to $6n+2$ was previously observed by Visu Makam in the talk~\cite{makam-talk}.

\subsubsection{Low-Degree Flattenings}

In the introduction, we argued why it might seem plausible to expect decomposition algorithms that reach rank $r \approx n^{3/2}$ for $n \times n \times n$ tensors with generic components. In the previous section we saw that linear flattenings cannot surpass rank $O(n)$ (and this was known already by the work of~\cite{barriers-rank}). One can imagine a much broader class of flattenings where the entries of $M(T)$ are higher-degree polynomials in the variables $T_{ijk}$, and this is stated as an open problem in Section~6 of~\cite{barriers-rank}. We will define such a class of strategies and show that even these will not surpass rank $O(n)$, where $O(\cdot)$ hides a constant depending on the polynomial degree.

For motivation, a common strategy for handling a third-order $n \times n \times n$ tensor is to apply a particular degree-2 polynomial map to produce a fourth-order $n \times n \times n \times n$ tensor, and then flatten this to an $n^2 \times n^2$ matrix. This trick is used, for instance, in work on refuting random constraint satisfaction problems~\cite{coja-refutation-ksat,refute-random-csp}, and within the sum-of-squares proof system it can be viewed as an application of the Cauchy--Schwarz inequality.

Inspired by this, here is a seemingly plausible approach for rank detection of third-order tensors. First, build some matrix $M(T)$ whose entries are degree-2 homogeneous polynomials in the variables $T_{ijk}$. Now we would like to ask for some version of ``rank additivity'' to hold, but since $T \mapsto M(T)$ is no longer linear, $M(\sum_\ell a^{(\ell)} \otimes b^{(\ell)} \otimes c^{(\ell)})$ no longer breaks down as the sum of the individual terms $M(a^{(\ell)} \otimes b^{(\ell)} \otimes c^{(\ell)})$. Instead, write $T = \sum_{\ell=1}^r T^{(\ell)}$ where $T^{(\ell)} = a^{(\ell)} \otimes b^{(\ell)} \otimes c^{(\ell)}$, and view the entries of $M(T)$ as (homogeneous degree-2) polynomials in the variables $T^{(\ell)}_{ijk}$. Decompose $M = M(T)$ as
\[ M = \sum_{1 \le \ell \le r} \hat{M}(T^{(\ell)}) + \sum_{1 \le \ell_1 < \ell_2 \le r} \tilde{M}(T^{(\ell_1)},T^{(\ell_2)}) =: M_\rep + M_\dis, \]
where $\hat{M}$ has entries that are homogeneous of degree 2 in its input tensor while $\tilde{M}$ has entries that are homogeneous with degree 1 in each of its two input tensors (so degree 2 in total). We have now written $M$ as the sum of $r$ ``repeat'' terms and $\binom{r}{2}$ ``distinct'' terms, and we might hope for the rank to be additive across these $\Theta(r^2)$ terms. It is tantalizing to imagine a scenario where $M$ has dimensions $N \times N$ with $N = \Theta(n^3)$, and each of the individual terms $\hat{M}(T^{(\ell)})$ or $\tilde{M}(T^{(\ell_1)},T^{(\ell_2)})$ has generic rank $\Theta(1)$, as this would potentially allow additive rank (and therefore successful rank detection) up to $r = \Theta(n^{3/2})$. Unfortunately we will see that this cannot pan out, and to prove a lower bound we will focus on only the ``distinct'' terms (which make up most of the terms) and show that even these alone cannot have additive rank.

We now define our setting more formally, and we will consider a generalization of the above discussion to any degree $d \ge 1$ and any tensor format $n_1 \times \cdots \times n_k$. We think of $d$ as a constant, so that the entries of $M$ can be computed in polynomial time.

\begin{definition}[Decomposition of $M$]
\label{def:M-decomp}
For $T \in \RR^{n_1 \times \cdots \times n_k}$, suppose $M = M(T)$ is an $N_1 \times N_2$ matrix whose entries are homogeneous degree-$d$ polynomials in the variables $T_{i_1,\ldots,i_k}$. Set $T = \sum_{\ell=1}^r T^{(\ell)}$ and, viewing the entries of $M(T)$ as (homogeneous degree-$d$) polynomials in the variables $T^{(\ell)}_{i_1,\ldots,i_k}$, consider the unique decomposition of $M = M(T)$ as
\[ M = M_\rep + \sum_{1 \le \ell_1 < \cdots < \ell_d \le r} \tilde{M}(T^{(\ell_1)},\ldots,T^{(\ell_d)}), \]
where $\tilde{M}$ has entries that are homogeneous with degree 1 in each of its $d$ input tensors (so degree $d$ in total), and $M_\rep$ consists of those monomials in $M$ that involve fewer than $d$ of the tensors $T^{(1)},\ldots,T^{(r)}$.
\end{definition}

\begin{theorem}\label{thm:lower-d}
Suppose $M$ is a matrix whose entries are homogeneous degree-$d$ polynomials in the entries of an $n_1 \times \cdots \times n_k$ tensor $T$. Set $T = \sum_{\ell=1}^r T^{(\ell)}$, and decompose $M(T)$ as in Definition~\ref{def:M-decomp}. Suppose that for $T^{(\ell)} = a^{(1,\ell)} \otimes \cdots \otimes a^{(k,\ell)}$ with generically chosen components $\{a^{(i,\ell)} \,:\, i \in [k], \, \ell \in [r]\}$, we have rank-$r$ additivity in the sense that
\[ \rank\left(\sum_{1 \le \ell_1 < \cdots < \ell_d \le r} \tilde{M}(T^{(\ell_1)},\ldots,T^{(\ell_d)})\right) = \binom{r}{d} \cdot \rank\left(\tilde{M}(T^{(1)},\ldots,T^{(d)})\right). \]
Then
\[ r \le 4^k d^2 \, n_* \]
where
\[ n_* := \max_{S \subseteq [k]} \, \min\left\{\prod_{i \in S} n_i, \, \prod_{i \in \bar{S}} n_i\right\}. \]
\end{theorem}

\noindent The proof can be found in Section~\ref{sec:pf-lower-d}. Recall from Section~\ref{sec:trivial} that the trivial flattening achieves rank detection up to rank $n_*$, so Theorem~\ref{thm:lower-d} shows that for any fixed $k,d$, a general class of higher-degree approaches cannot beat the trivial flattening by more than a constant factor $4^k d^2$.

\begin{remark}
We discuss a few sanity checks on the definition of rank-$r$ additivity used in Theorem~\ref{thm:lower-d}. First note that for $d=1$, this definition reduces to the one from~\eqref{eq:def-additivity}, since $M_\rep = 0$. Next, let us give an example of a flattening that satisfies rank-$r$ additivity for arbitrary $d$. Start with any linear (homogeneous degree-1) flattening $T \mapsto L(T)$ and consider the flattening $M(T) = L(T)^{\otimes d}$ (the $d$-fold Kronecker product), which is homogeneous of degree $d$. Suppose that (for generic components), $L$ has rank-$r$ additivity in the sense of~\eqref{eq:def-additivity}. We claim that $M$ has rank-$r$ additivity in the sense of Theorem~\ref{thm:lower-d} (again, for generic components). We sketch the proof in the special case $k=3$, $d=2$, but the proof can be extended to all $k$ and $d$. First, a basic property of the Kronecker product is $\rank(A \otimes B) = \rank(A) \cdot \rank(B)$. When $T = \sum_{\ell=1}^r T^{(\ell)} = \sum_{\ell=1}^r a^{(\ell)} \otimes b^{(\ell)} \otimes c^{(\ell)}$, we can decompose $M(T) = \sum_{\ell_1,\ell_2 \in [r]} L(T^{(\ell_1)}) \otimes L(T^{(\ell_2)})$. Using the additivity for $L$, the $r^2$ terms in this sum must have additive rank. This implies additivity for the $\binom{r}{2}$ ``distinct'' terms, which in this case are $\tilde{M}(T^{(\ell_1)},T^{(\ell_2)}) = L(T^{(\ell_1)}) \otimes L(T^{(\ell_2)}) + L(T^{(\ell_2)}) \otimes L(T^{(\ell_1)})$ for $1 \le \ell_1 < \ell_2 \le r$.
\end{remark}

\subsection{Open Problems}

We suggest the following questions as interesting directions for future work.

\begin{itemize}
    \item For an $n \times n \times n$ tensor with generic components, what is the maximum rank for which a linear flattening can have additivity? We have shown the answer is asymptotically $cn$ for a constant $c \in [2,6]$. This improved the previous range of $[3/2,8]$ by~\cite{persu-thesis} and~\cite{barriers-rank}, but the exact value of $c$ remains unknown.
    \item Can rank detection improve with the degree of flattening? Our lower bounds do not rule out this possibility, but it is possible the Koszul--Young flattening is already optimal.
    \item For order-4 tensors of format $n \times n \times n \times n$ (with generic components), the trivial flattening gives rank detection up to rank $n^2$, but existing decomposition results stop at $cn^2$ for a constant $c < 1$~\cite{foobi,MSS,HSS-robust,JLV}. Can we improve the decomposition results to rank $(1-\epsilon)n^2$? This might require improving the constant $1/4$ in~\cite{JLV}, which may be sub-optimal (but curiously was not a bottleneck in the order-3 case). We remark that the constant has been improved from $1/4$ to $1-\epsilon$ for a non-planted variant of the problem~\cite{X-arability}.
    \item Can our algorithmic guarantees be extended to \emph{symmetric} tensors of format $n \times n \times n$?
\end{itemize}

\section{Proofs for Upper Bounds}

\subsection{Proof of Proposition~\ref{prop:rank}}
\label{sec:pf-prop-rank}

The genericity assumptions we will need on the components are that neither $b$ nor $c$ is the zero vector and that $a_1,a_2,\ldots,a_q$ are not all zero. Recall $M = A \otimes (bc^\top)$, so $\rank(M) = \rank(A)$. Now Proposition~\ref{prop:rank} follows immediately from Lemma~\ref{lem:rank-A} below, which shows that $A$ generically has rank $\binom{q-1}{p}$. (Lemma~\ref{lem:rank-A} also includes additional facts about the structure of $A$ that will be needed later.)

\begin{lemma}\label{lem:rank-A}
Consider the matrix $A = A(a;p,q)$ defined in~\eqref{eq:def-A}. If $a_1,\ldots,a_q$ are all zero then $A = 0$ and so $\rank(A) = 0$. Otherwise, $\rank(A) = \binom{q-1}{p}$. Furthermore, if $a_i \ne 0$ for some $i \in [q]$ then the columns of $A$ indexed by $\{U \,:\, i \in U\}$ form a basis for the column span of $A$, and similarly the rows $\{S \,:\, i \notin S\}$ form a basis for the row span.
\end{lemma}

\begin{proof}
Fix $i \in [q]$ such that $a_i \ne 0$. We first show $\rank(A) \ge \binom{q-1}{p}$ by demonstrating a $\binom{q-1}{p} \times \binom{q-1}{p}$ submatrix with full rank. Specifically, consider the submatrix indexed by $\{S \,:\, i \notin S\} \times \{U \,:\, i \in U\}$. Note that this submatrix has exactly one nonzero entry per row and exactly one nonzero entry per column, and all the nonzero entries are equal to $\pm a_i$. This shows $\rank(A) \ge \binom{q-1}{p}$.

The full-rank submatrix demonstrated above implies that the columns $\{U \,:\, i \in U\}$ are linearly independent, and same for the columns $\{S \,:\, i \notin S\}$. It remains to show $\rank(A) \le \binom{q-1}{p}$. To do this, we will show that every column of $A$ lies in the span of the $\binom{q-1}{p}$ columns $\{U \,:\, i \in U\}$. Fix a column $W$ with $i \notin W$, and let $Y = W \sqcup \{i\}$. We claim that the columns indexed by $\{Y \setminus \{j\} \,:\, j \in Y\}$ are linearly dependent, where the coefficients are $\sigma(Y,j) \, a_j$. To see this, consider entry $S$ of this linear combination. If $S \not\subseteq Y$ then entry $S$ is equal to zero, and otherwise $S = Y \setminus \{j,k\}$ for some $j,k \in Y$ so entry $S$ is equal to
\[ \sigma(Y,j) \, a_j \, A_{S,Y \setminus \{j\}} + \sigma(Y,k) \, a_k \, A_{S,Y \setminus \{k\}} = \sigma(Y,j) \, a_j \, \sigma(Y \setminus \{j\},k) \, a_k + \sigma(Y,k) \, a_k \, \sigma(Y \setminus \{k\},j) \, a_j, \]
which is zero because $\sigma(Y,j) \, \sigma(Y \setminus \{j\},k)$ and $\sigma(Y,k) \, \sigma(Y \setminus \{k\},j)$ have opposite signs. Now take the linear dependence from above and solve for column $W$ (which has nonzero coefficient $\sigma(Y,i) \, a_i$); this allows column $W$ to be expressed as a linear combination of the columns $\{U \,:\, i \in U\}$.
\end{proof}

\subsection{Proof of Theorem~\ref{thm:rank-det}}
\label{sec:pf-rank-det}

Provided $a_i \ne 0$ for all $i \in [q]$ (which is true for generic $a$), the matrix $A = A(a;p,q)$ defined in~\eqref{eq:def-A} can be factored as $A = \diag(v) \cdot \tilde{A} \cdot \diag(w)$ where $\tilde{A}_{SU} = \sum_{i \in [q]} \One_{U = S \sqcup \{i\}} \cdot \sigma(U,i)$, $v_S = \prod_{i \in S} a_i^{-1}$, and $w_U = \prod_{i \in U} a_i$. Note that $\tilde{A} = A(\One;p,q)$ where $\One := (1,\ldots,1)^\top$, so by Lemma~\ref{lem:rank-A} we have $\rank(\tilde{A}) = \binom{q-1}{p}$. We can therefore factor $\tilde{A} = \tilde{Q} \tilde{R}^\top$ where $\tilde{Q},\tilde{R}$ are $\binom{q}{p} \times \binom{q-1}{p}$ and $\binom{q}{p+1} \times \binom{q-1}{p}$, respectively. Now write $A = QR^\top$ where $Q = \diag(v) \cdot \tilde{Q}$ and $R = \diag(w) \cdot \tilde{R}$. We use the superscript $(\ell)$ to mark that a matrix is built from the tensor component $a^{(\ell)}$, so for instance, $A^{(\ell)} := A(a^{(\ell)};p,q)$. Factor $M$ as
\begin{equation}\label{eq:M-factor}
M = \sum_{\ell=1}^r A^{(\ell)} \otimes (b^{(\ell)} c^{(\ell) \top}) = \left[\begin{array}{ccc} Q^{(1)} & & Q^{(r)} \\ \otimes & \cdots & \otimes \\ b^{(1)} & & b^{(r)} \end{array} \right] \left[\begin{array}{ccc} R^{(1)} & & R^{(r)} \\ \otimes & \cdots & \otimes \\ c^{(1)} & & c^{(r)} \end{array} \right]^\top
\end{equation}
where, e.g., $Q^{(1)} \otimes b^{(1)}$ denotes the $\binom{q}{p} n_2 \times \binom{q-1}{p}$ Kronecker product. This gives a factorization of $M$ as the product of two matrices with inner dimension $r\binom{q-1}{p}$, so it suffices to show that the two factors on the right-hand side of~\eqref{eq:M-factor} have full column rank and full row rank, respectively. We will show this for the first factor and then explain the (minor) changes needed to adapt the proof for the second factor.

Let $m$ be the smallest integer for which $r \le mq$. It suffices to verify the extreme case $r = mq$, i.e., we will show that even when $r$ is increased to add additional columns, all the columns are linearly independent. Focusing on the first factor in~\eqref{eq:M-factor}, the first $r \binom{q-1}{p}$ rows --- call this square submatrix $Q'$ --- constitute an $m \times m$ grid of square blocks, where each square block has dimension $(q-p)\binom{q}{p} = q\binom{q-1}{p}$ consisting of (scaled) copies of $Q$ arranged in a $(q-p) \times q$ grid. We require $n_2 \ge m(q-p)$ so that there are enough rows available to form the submatrix $Q'$.

Our goal is to show that $\det(Q')$ is generically nonzero, and it suffices to show this after plugging in values for some of the $b$ variables. Namely, by setting the appropriate $b$ variables to zero, we can keep only the $m$ square blocks on the diagonal of $Q'$, and thus it suffices to show invertibility for a single square block of the form
\[ Q'' = \left[ \begin{array}{cccc} b_{11} Q^{(1)} & b_{12} Q^{(2)} 
& \cdots & b_{1,q} Q^{(q)} \\ b_{21} Q^{(1)} & b_{22} Q^{(2)} & & \vdots \\ \vdots & & \ddots & \\ b_{q-p,1} Q^{(1)} & \cdots & & b_{q-p,q} Q^{(q)} \end{array} \right], \]
where $b_{j\ell} := b^{(\ell)}_j$, and we will similarly write $a_{i\ell} := a^{(\ell)}_i$. Partition $Q''$ into ``cells'' of size $1 \times \binom{q-1}{p}$: a cell is indexed by $(S,j,\ell)$ where $S \subseteq [q]$, $|S| = p$, $j \in [q-p]$, and $\ell \in [q]$, and the content of the $(S,j,\ell)$ cell is $(b_{j\ell} \prod_{i \in S} a_{i\ell}^{-1}) \tilde{Q}_S$ where $\tilde{Q}_S$ denotes row $S$ of $\tilde{Q}$.

We aim to show $\det(Q'')$ is nonzero as a polynomial in the variables $b_{j\ell}, a_{i\ell}^{-1}$. In the usual expansion of $\det(Q'')$ as a sum over permutations, each term chooses exactly 1 entry per row and column of $Q''$. This choice can be broken down as first choosing a set of cells $(S,j,\ell)$ with exactly 1 cell per row $(S,j)$ and exactly $\binom{q-1}{p}$ cells per ``pillar'' $\ell$ (where ``pillar'' describes a vertical stack of cells that are all indexed by the same $\ell$ value), and then choosing one entry per cell in a way that each column appears once. Let $H$ denote a subset of tuples $(S,j,\ell)$ representing the choice at the first stage. For a given $H$, define matrices $F_1^H,\ldots,F_q^H$, each of size $\binom{q-1}{p} \times \binom{q-1}{p}$, where $F_\ell^H$ is formed by collecting the rows $\tilde{Q}_S$ for each $(S,j,\ell) \in H$ (i.e., the cells selected in pillar $\ell$). Now
\begin{equation}\label{eq:det-expand}
\det(Q'') = \sum_H \tau_H \cdot (a,b)^H \cdot \prod_{\ell=1}^q \det(F_\ell^H),
\end{equation}
where $(a,b)^H$ denotes the monomial $\prod_{(S,j,\ell) \in H} b_{j\ell} \prod_{i \in S} a_{i\ell}^{-1}$, and $\tau_H \in \{\pm 1\}$ are signs whose values will not matter to us. If $F_\ell^H$ has a repeated row $\tilde{Q}_S$ then its determinant is 0, so we need only consider $H$ for which the cells in each pillar have distinct $S$ values.

To complete the proof, we will identify a particular monomial that has nonzero coefficient in $\det(Q'')$. To this end, define a ``good'' set of cells $G$ as follows, which (like $H$) will have 1 cell per row and $\binom{q-1}{p}$ cells per pillar. Within each pillar $\ell \in [q]$, for each $S$ with $\ell \notin S$, include the cell $(S,j,\ell)$ where $j$ is the first available index (i.e., the minimum $j \in [q-p]$ such that $(S,j,\ell')$ has not already been declared a good cell for some $\ell' < \ell$). Each $S$ is used $q-p$ times in total (across all pillars), so the $j$ values will not ``run out.'' We claim that the only term in~\eqref{eq:det-expand} that produces the monomial $(a,b)^G$ (with nonzero coefficient) is $H = G$: within each pillar $\ell = 1,2,\ldots$, the powers of $a$ (particularly the absence of $a_{\ell\ell}^{-1}$) determine the collection of $S$ values that $H$ must use (namely one copy of each $S$ for which $\ell \notin S$), and the powers of $b$ determine the multiset of $j$ values to use. The choice of ``first available $j$'' used in $G$ constrains a unique pairing between $S$ and $j$ values.

We conclude that the coefficient of $(a,b)^G$ in $\det(Q'')$ is $\pm \prod_{\ell=1}^q \det(F_\ell^G)$ and it remains to show that each $\det(F_\ell^G)$ is nonzero. This follows because $F_\ell^G$ has rows $\tilde{Q}_S$ for all $S$ with $\ell \notin S$, and the corresponding rows of $\tilde{A} = A(\One) = \tilde{Q} \tilde{R}^\top$ are linearly independent by Lemma~\ref{lem:rank-A}.

To handle the second factor in~\eqref{eq:M-factor}, we also need to show invertibility of the analogous square block
\[ R'' = \left[ \begin{array}{cccc} c_{11} R^{(1)} & c_{12} R^{(2)} 
& \cdots & c_{1,q} R^{(q)} \\ c_{21} R^{(1)} & c_{22} R^{(2)} & & \vdots \\ \vdots & & \ddots & \\ c_{p+1,1} R^{(1)} & \cdots & & c_{p+1,q} R^{(q)} \end{array} \right], \]
whose entries are polynomials in the variables $a_{i\ell} := a^{(\ell)}_i$ and $c_{j\ell} := c^{(\ell)}_j$. We again partition $R''$ into cells of size $1 \times \binom{q-1}{p}$ indexed by $(U,j,\ell)$ where $U \subseteq [q]$, $|U| = p+1$, $j \in [p+1]$, and $\ell \in [q]$, and the content of the $(U,j,\ell)$ cell is $(c_{j\ell} \prod_{i \in U} a_{i\ell})\tilde{R}_U$ where $\tilde{R}_U$ denotes row $U$ of $\tilde{R}$. The remainder of the argument is essentially identical to above, with one minor change: in the construction of $G$, each pillar $\ell$ includes the $U$ values for which $\ell \in U$ (in place of the previous condition $\ell \notin S$).

Finally, the above constructions required an integer $m$ such that $r \le mq$, $n_2 \ge m(q-p)$, and $n_3 \ge m(p+1)$, or equivalently,
\[ \frac{r}{q} \le m \le \min\left\{\frac{n_2}{q-p}, \, \frac{n_3}{p+1}\right\}. \]
For this, it suffices to have
\begin{equation}\label{eq:before-sub-p}
r \le q\left(\min\left\{\frac{n_2}{q-p}, \, \frac{n_3}{p+1}\right\} - 1\right).
\end{equation}
Now using the choice of $p = \lfloor q \cdot n_3/(n_2+n_3) \rfloor$, we can bound
\begin{equation}\label{eq:p-bound-1}
\frac{q n_2}{q-p} \ge \frac{q n_2}{q-\left(q \cdot \frac{n_3}{n_2+n_3} - 1\right)} = \frac{n_2+n_3}{1+\frac{1}{q}\left(1+\frac{n_3}{n_2}\right)} \ge (n_2+n_3)\left[1-\frac{1}{q}\left(1+\frac{n_3}{n_2}\right)\right]
\end{equation}
and
\begin{equation}\label{eq:p-bound-2}
\frac{qn_3}{p+1} \ge \frac{qn_3}{q \cdot \frac{n_3}{n_2+n_3}+1} = \frac{n_2+n_3}{1+\frac{1}{q}\left(1+\frac{n_2}{n_3}\right)} \ge (n_2+n_3) \left[1 - \frac{1}{q}\left(1+\frac{n_2}{n_3}\right)\right].
\end{equation}
Therefore it suffices to have
\[ r \le (n_2+n_3)\left[1-\frac{1}{q}\left(1+\max\left\{\frac{n_2}{n_3},\frac{n_3}{n_2}\right\}\right)\right] - q. \]

\subsection{Proof of Theorem~\ref{thm:uniqueness}}
\label{sec:pf-uniqueness}

Suppose $T$ has an alternative decomposition $T = \sum_{\ell=1}^r \tilde{a}^{(\ell)} \otimes \tilde{b}^{(\ell)} \otimes \tilde{c}^{(\ell)}$. Recalling the setup in Section~\ref{sec:rank-det},
\[ M(T;p,q) = \sum_{\ell=1}^r A(a^{(\ell)};p,q) \otimes (b^{(\ell)}c^{(\ell)\top}) = \sum_{\ell=1}^r A(\tilde{a}^{(\ell)};p,q) \otimes (\tilde{b}^{(\ell)}\tilde{c}^{(\ell)\top}). \]
By assumption~\ref{item:M}, this matrix has rank exactly $r\binom{q-1}{p}$. By Lemma~\ref{lem:rank-A}, the matrix $A(a;p,q)$ always has rank at most $\binom{q-1}{p}$, which means each term $A(a^{(\ell)};p,q) \otimes (b^{(\ell)}c^{(\ell)\top})$ must have rank exactly $\binom{q-1}{p}$, and in particular, $b^{(\ell)} \ne 0$ and $c^{(\ell)} \ne 0$ (i.e., for each $\ell$, neither $b^{(\ell)}$ nor $c^{(\ell)}$ is the zero vector). Furthermore, every column of $A(a^{(\ell)};p,q) \otimes (b^{(\ell)}c^{(\ell)\top})$ must lie in the column span of $M$, and so
\begin{equation}\label{eq:colspan-M}
\Colspan(M) = \Span\{z^{(U)}(a^{(\ell)};p,q) \otimes b^{(\ell)} \,:\, U \subseteq [q], \, |U| = p+1, \, \ell \in [r]\},
\end{equation}
where, recall, $z^{(U)}(a;p,q)$ denotes column $U$ of $A(a;p,q)$. Similarly, we have the same conclusion for the alternative decomposition: $\tilde{b}^{(\ell)} \ne 0$ and $\tilde{c}^{(\ell)} \ne 0$ for each $\ell \in [r]$, and
\begin{equation}\label{eq:colspan-M-alt}
\Colspan(M) = \Span\{z^{(U)}(\tilde{a}^{(\ell)};p,q) \otimes \tilde{b}^{(\ell)} \,:\, U \subseteq [q], \, |U| = p+1, \, \ell \in [r]\}.
\end{equation}
Now recall that we define $V := [p+1]$ and denote by $Z_{p,q} \subseteq \RR^{\binom{q}{p}}$ the subspace consisting of vectors with the same sparsity pattern as $z^{(V)}$, that is, nonzeros are only allowed in the $p+1$ entries $S$ where $S \subseteq V$.

\begin{lemma}\label{lem:subspace-intersect}
Under the assumptions of Theorem~\ref{thm:uniqueness},
\[ \Colspan(M) \cap (Z_{p,q} \otimes \RR^{n_2}) = \Span\{z^{(V)}(a^{(\ell)};p,q) \otimes b^{(\ell)} \,:\, \ell \in [r]\}. \]
\end{lemma}

\noindent The proof of Lemma~\ref{lem:subspace-intersect} is deferred to Section~\ref{sec:pf-subspace-intersect}. The proof relies on assumption~\ref{item:N}.

Recall that $\varphi_{p,q}$, defined in~\eqref{eq:def-phi}, is an isomorphism $Z_{p,q} \to \RR^{p+1}$ designed so that $\varphi_{p,q}(z^{(V)}(a;p,q)) = (a_1,\ldots,a_{p+1})^\top$. From Lemma~\ref{lem:subspace-intersect} we now conclude
\begin{equation}\label{eq:apply-phi-2}
(\varphi_{p,q} \otimes I)(\Colspan(M) \cap (Z_{p,q} \otimes \RR^{n_2})) = \Span\{d^{(\ell)} \otimes b^{(\ell)} \,:\, \ell \in [r]\},
\end{equation}
where, recall, $d^{(\ell)} := (a^{(\ell)}_1,\ldots,a^{(\ell)}_{p+1})^\top$.

\begin{lemma}\label{lem:JLV}
Under the assumptions of Theorem~\ref{thm:uniqueness}, the vectors $\{d^{(\ell)} \otimes b^{(\ell)} \,:\, \ell \in [r]\}$ are the only rank-1 tensors, up to scalar multiple, in their span, $\Span\{d^{(\ell)} \otimes b^{(\ell)} \,:\, \ell \in [r]\}$. Furthermore, the algorithm of~\cite[Corollary~3]{JLV} recovers (scalar multiples of) these $r$ rank-1 tensors given their span.
\end{lemma}

\noindent The proof of Lemma~\ref{lem:JLV} is based on~\cite{JLV} and can be found in Section~\ref{sec:pf-JLV}. The proof relies on assumptions~\ref{item:db-li} and~\ref{item:P}.

Define $\tilde{d}^{(\ell)} := (\tilde{a}^{(\ell)}_1,\ldots,\tilde{a}^{(\ell)}_{p+1})^\top$. From~\eqref{eq:colspan-M-alt} we have $z^{(V)}(\tilde{a}^{(\ell)};p,q) \otimes \tilde{b}^{(\ell)} \in \Colspan(M) \cap (Z_{p,q} \otimes \RR^{n_2})$ for all $\ell \in [r]$, and so $\tilde{d}^{(\ell)} \otimes \tilde{b}^{(\ell)} \in (\varphi_{p,q} \otimes I)(\Colspan(M) \cap (Z_{p,q} \otimes \RR^{n_2}))$. But by~\eqref{eq:apply-phi-2} and Lemma~\ref{lem:JLV}, the only rank-1 tensors in this subspace are scalar multiples of $d^{(\ell)} \otimes b^{(\ell)}$. We therefore conclude that for every $\ell \in [r]$ there exists $\ell' \in [r]$ such that $\tilde{d}^{(\ell)} \otimes \tilde{b}^{(\ell)}$ is a scalar multiple of $d^{(\ell')} \otimes b^{(\ell')}$.

Next we need to repeat the argument with the ``$b$'' and ``$c$'' modes switched and with $p$ replaced by $q-p-1$, which will use the assumptions~\ref{item:dc-li}, \ref{item:M-prime}, \ref{item:N-prime}, \ref{item:P-prime}. Due to symmetry, the argument is identical to before, and we only state the main claims. First, the flattening $M'$ takes the form
\[ M'(T;p,q) = \sum_{\ell=1}^r A(a^{(\ell)};q-p-1,q) \otimes (c^{(\ell)}b^{(\ell)\top}) = \sum_{\ell=1}^r A(\tilde{a}^{(\ell)};q-p-1,q) \otimes (\tilde{c}^{(\ell)}\tilde{b}^{(\ell)\top}), \]
and by Lemma~\ref{lem:rank-A}, the maximum possible rank of $A(a;q-p-1,q)$ is $\binom{q-1}{q-p-1} = \binom{q-1}{p}$. Analogous to~\eqref{eq:colspan-M} and~\eqref{eq:colspan-M-alt}, we have
\[ \Colspan(M') = \Span\{z^{(U)}(a^{(\ell)};q-p-1,q) \otimes c^{(\ell)} \,:\, U \subseteq [q], \, |U| = q-p, \, \ell \in [r]\} \]
and
\[ \Colspan(M') = \Span\{z^{(U)}(\tilde{a}^{(\ell)};q-p-1,q) \otimes \tilde{c}^{(\ell)} \,:\, U \subseteq [q], \, |U| = q-p, \, \ell \in [r]\}. \]
Analogous to Lemma~\ref{lem:subspace-intersect} and~\eqref{eq:apply-phi-2}, we have, with $V' := [q-p]$,
\[ \Colspan(M') \cap (Z_{q-p-1,q} \otimes \RR^{n_3}) = \Span\{z^{(V')}(a^{(\ell)};q-p-1,q) \otimes c^{(\ell)} \,:\, \ell \in [r]\}. \]
and
\[ (\varphi_{q-p-1,q} \otimes I)(\Colspan(M') \cap (Z_{q-p-1,q} \otimes \RR^{n_3})) = \Span\{f^{(\ell)} \otimes c^{(\ell)} \,:\, \ell \in [r]\}, \]
where, recall, $f^{(\ell)} := (a^{(\ell)}_1,\ldots,a^{(\ell)}_{q-p})^\top$. Analogous to Lemma~\ref{lem:JLV}, the vectors $\{f^{(\ell)} \otimes c^{(\ell)} \,:\, \ell \in [r]\}$ are the only rank-1 tensors in their span, up to scalar multiple, and the algorithm of~\cite{JLV} recovers them from their span. Finally, defining $\tilde{f}^{(\ell)} := (\tilde{a}^{(\ell)}_1,\ldots,\tilde{a}^{(\ell)}_{q-p})^\top$, we have $\tilde{f}^{(\ell)} \otimes \tilde{c}^{(\ell)} \in (\varphi_{q-p-1,q} \otimes I)(\Colspan(M') \cap (Z_{q-p-1,q} \otimes \RR^{n_3}))$ and so every $\tilde{f}^{(\ell)} \otimes \tilde{c}^{(\ell)}$ is a scalar multiple of some $f^{(\ell'')} \otimes c^{(\ell'')}$.

Let $R \subseteq [r]$ be the set of $\ell$ for which $\tilde{a}^{(\ell)}_1 \ne 0$. For each $\ell \in [r]$ we have $\tilde{d}^{(\ell)} \ne 0$ and $\tilde{f}^{(\ell)} \ne 0$. Recall that $b^{(\ell)}, \tilde{b}^{(\ell)}, c^{(\ell)}, \tilde{c}^{(\ell)}, d^{(\ell)}, f^{(\ell)}$ are all nonzero (where $d^{(\ell)} \ne 0$ since $a^{(\ell)}_1 \ne 0$ by assumption~\ref{item:a1}, and similarly for $f^{(\ell)}$). For each $\ell \in R$ we have established that $\tilde{d}^{(\ell)} \otimes \tilde{b}^{(\ell)}$ is a nonzero scalar multiple of some $d^{(\ell')} \otimes b^{(\ell')}$, and $\tilde{f}^{(\ell)} \otimes \tilde{c}^{(\ell)}$ is a nonzero scalar multiple of some $f^{(\ell'')} \otimes c^{(\ell'')}$. In particular, $\tilde{d}^{(\ell)}$ is a nonzero scalar multiple of $d^{(\ell')}$, and $\tilde{f}^{(\ell)}$ is a nonzero scalar multiple of $f^{(\ell'')}$. Since $\tilde{d}^{(\ell)}$ and $\tilde{f}^{(\ell)}$ agree on the first $\bar{p}$ coordinates (and similarly for $d^{(\ell)},f^{(\ell)}$), Assumption~\ref{item:d-distinct} implies that $\ell' = \ell''$. This means we can define a function $\pi: R \to [r]$ such that $\tilde{b}^{(\ell)} \otimes \tilde{c}^{(\ell)} = \eta_\ell \cdot b^{(\pi(\ell))} \otimes c^{(\pi(\ell))}$ for a scalar $\eta_\ell \ne 0$, for each $\ell \in R$.

Let $T_1 := (T_{1jk} \,:\, j \in [n_2], \, k \in [n_3])$ denote the first $n_2 \times n_3$ matrix slice of $T$, namely
\[ T_1 = \sum_{\ell=1}^r a^{(\ell)}_1 \cdot b^{(\ell)} \otimes c^{(\ell)}. \]
Similarly for the alternative decomposition,
\[ T_1 = \sum_{\ell=1}^r \tilde{a}^{(\ell)}_1 \cdot \tilde{b}^{(\ell)} \otimes \tilde{c}^{(\ell)} = \sum_{\ell \in R} \tilde{a}^{(\ell)}_1 \cdot \eta_\ell \cdot b^{(\pi(\ell))} \otimes c^{(\pi(\ell))}, \]
since $\tilde{a}^{(\ell)}_1 = 0$ for $\ell \notin R$. By assumptions~\ref{item:a1} and \ref{item:bc-li}, $a^{(\ell)}_1 \ne 0$ for all $\ell \in [r]$, and the vectors $\{b^{(\ell)} \otimes c^{(\ell)} \,:\, \ell \in [r]\}$ are linearly independent. Comparing the previous expressions for $T_1$, this means $R = [r]$, $\pi$ is a permutation of $[r]$, and $\tilde{a}^{(\ell)}_1 = \eta_\ell^{-1} \cdot a^{(\pi(\ell))}_1$. Now for each slice $T_i := (T_{ijk} \,:\, j \in [n_2], \, k \in [n_3]\}$ with $i \in [n_1]$, we can write
\[ T_i = \sum_{\ell=1}^r a^{(\ell)}_i \cdot b^{(\ell)} \otimes c^{(\ell)} = \sum_{\ell=1}^r \tilde{a}^{(\ell)}_i \cdot \eta_\ell \cdot b^{(\pi(\ell))} \otimes c^{(\pi(\ell))}, \]
and again using the linear independence of $\{b^{(\ell)} \otimes c^{(\ell)} \,:\, \ell \in [r]\}$ we must have $\tilde{a}^{(\ell)}_i = \eta_\ell^{-1} \cdot a^{(\pi(\ell))}_i$. This means $\tilde{a}^{(\ell)} = \eta_\ell^{-1} \cdot a^{(\pi(\ell))}$ and so the alternative decomposition has the same rank-1 terms as the original: $\tilde{a}^{(\ell)} \otimes \tilde{b}^{(\ell)} \otimes \tilde{c}^{(\ell)} = a^{(\pi(\ell))} \otimes b^{(\pi(\ell))} \otimes c^{(\pi(\ell))}$, with $\pi$ a permutation of $[r]$. This completes the proof of uniqueness.

The above argument also implies success of Algorithm~\ref{alg:decomp}. From~\eqref{eq:apply-phi-2} we know the subspace $B$ constructed by the algorithm is equal to $\Span\{d^{(\ell)} \otimes b^{(\ell)} \,:\, \ell \in [r]\}$, and similarly, $C = \Span\{f^{(\ell)} \otimes c^{(\ell)} \,:\, \ell \in [r]\}$. By Lemma~\ref{lem:JLV}, the rank-1 tensors $\hat{d}^{(\ell)} \otimes \hat{b}^{(\ell)}$ recovered by the~\cite{JLV} algorithm satisfy $\hat{d}^{(\ell)} = \alpha_\ell \cdot d^{(\pi_1(\ell))}$ and $\hat{b}^{(\ell)} = \beta_\ell \cdot b^{(\pi_1(\ell))}$ for scalars $\alpha_\ell \ne 0,\, \beta_\ell \ne 0$ and a permutation $\pi_1$ of $[r]$. Similarly, $\hat{f}^{(\ell)} = \gamma_\ell \cdot f^{(\pi_2(\ell))}$ and $\hat{c}^{(\ell)} = \delta_\ell \cdot c^{(\pi_2(\ell))}$ for scalars $\gamma_\ell \ne 0, \, \delta_\ell \ne 0$ and a permutation $\pi_2$ of $[r]$. Using Assumption~\ref{item:d-distinct}, the (unique) permutation $\tau$ found by Algorithm~\ref{alg:decomp} is $\tau = \pi_2^{-1} \circ \pi_1$. The final step is to solve the linear system of equations~\eqref{eq:linear} for $w^{(\ell)}$, which reduces to
\[ \sum_{\ell=1}^r w^{(\ell)}_i \cdot (\beta_\ell \cdot b^{(\pi_1(\ell))}) \otimes (\delta_{\tau(\ell)} \cdot c^{(\pi_1(\ell))}) = T_i \qquad \text{for } i \in [n_1]. \]
One solution is $w^{(\ell)} = \beta_\ell^{-1} \cdot \delta_{\tau(\ell)}^{-1} \cdot a^{(\pi_1(\ell))}$, which recovers the desired decomposition. This solution is unique because $\{b^{(\ell)} \otimes c^{(\ell)} \,:\, \ell \in [r]\}$ are linearly independent by assumption~\ref{item:bc-li}.

\subsection{Proof of Lemma~\ref{lem:subspace-intersect}}
\label{sec:pf-subspace-intersect}

Using~\eqref{eq:colspan-M}, the vectors $z^{(V)}(a^{(\ell)};p,q) \otimes b^{(\ell)}$ for $\ell \in [r]$ lie in $\Colspan(M) \cap (Z_{p,q} \otimes \RR^{n_2})$, so the inclusion ``$\supseteq$'' in Lemma~\ref{lem:subspace-intersect} holds. For the reverse inclusion, assume on the contrary that ``$\subseteq$'' fails. By assumption~\ref{item:a1} we have $a^{(\ell)}_1 \ne 0$ for all $\ell \in [r]$, so using~\eqref{eq:colspan-M} and Lemma~\ref{lem:rank-A}, the vectors $\{z^{(U)}(a^{(\ell)};p,q) \otimes b^{(\ell)} \,:\, U \subseteq [q], \, |U| = p+1, \, 1 \in U, \, \ell \in [r]\}$ span $\Colspan(M)$; in fact, since the number of these vectors matches $\rank(M)$, they must form a basis for $\Colspan(M)$. Define $\cU_1 := \{U \subseteq [q] \,:\, |U| = p+1, \, 1 \in U, \, U \ne V\}$. Now, failure of ``$\subseteq$'' is equivalent to the existence of coefficients $\{\gamma_{U,\ell} \in \RR \,:\, U \in \cU_1, \, \ell \in [r]\}$, not all zero, such that
\begin{equation}\label{eq:gamma-coeff}
\sum_{U \in \cU_1, \, \ell \in [r]} \gamma_{U,\ell} \cdot (z^{(U)}(a^{(\ell)};p,q) \otimes b^{(\ell)}) \in Z_{p,q} \otimes \RR^{n_2}.
\end{equation}
Membership in $Z_{p,q} \otimes \RR^{n_2}$ amounts to having a particular sparsity pattern, namely the entries $(S,j)$ with $S \not\subseteq V = [p+1]$ must be zero. This means we can write~\eqref{eq:gamma-coeff} in matrix form as $N\gamma = 0$ where $N$ is defined in~\eqref{eq:def-N}. But $N$ has full column rank by assumption~\ref{item:N}, a contradiction.

\subsection{Proof of Lemma~\ref{lem:JLV}}
\label{sec:pf-JLV}

The following arguments are based on~\cite{JLV}. We will unpack some details of~\cite{JLV} in order to extract an explicit condition under which their algorithm succeeds.

Define the set of rank-1 tensors,
\[ \cX := \{d \otimes b \,:\, d \in \RR^{p+1}, \, b \in \RR^{n_2}\}, \]
and the subspace
\[ \cY := \Span\{d^{(\ell)} \otimes b^{(\ell)} \,:\, \ell \in [r]\}. \]
For any $z \in \cX \cap \cY$, the vector $w := z \otimes z$ must belong to the subspace $\Span\{y \otimes y \,:\, y \in \cY\}$. Since a basis for $\cY$ is given by $y^{(1)},\ldots,y^{(r)}$ with $y^{(\ell)} := d^{(\ell)} \otimes b^{(\ell)}$ (which are linearly independent by assumption~\ref{item:db-li}), a basis for the subspace $\Span\{y \otimes y \,:\, y \in \cY\}$ is given by $\{y^{(\ell_1)} \otimes y^{(\ell_2)} + y^{(\ell_2)} \otimes y^{(\ell_1)} \,:\, 1 \le \ell_1 \le \ell_2 \le r\}$, so we can write
\begin{equation}\label{eq:delta-1}
w = \sum_{1 \le \ell_1 \le \ell_2 \le r} \delta_{\ell_1,\ell_2} (d^{(\ell_1)} \otimes b^{(\ell_1)} \otimes d^{(\ell_2)} \otimes b^{(\ell_2)} + d^{(\ell_2)} \otimes b^{(\ell_2)} \otimes d^{(\ell_1)} \otimes b^{(\ell_1)}),
\end{equation}
for some coefficients $\delta_{\ell_1,\ell_2} \in \RR$. Also, viewing $z$ as a $(p+1) \times n_2$ matrix, since $z$ has rank 1 we know every 2-by-2 minor must vanish:
\[ z_{i_1,j_1} z_{i_2,j_2} - z_{i_1,j_2} z_{i_2,j_1} = 0 \qquad \text{for } 1 \le i_1 < i_2 \le p+1 \text{ and } 1 \le j_1 < j_2 \le n_2. \]
Equivalently,
\begin{equation}\label{eq:delta-2}
\langle w, E_{i_1,i_2;j_1,j_2} \rangle = 0 \qquad \text{for } 1 \le i_1 < i_2 \le p+1 \text{ and } 1 \le j_1 < j_2 \le n_2,
\end{equation}
where
\[ E_{i_1,i_2;j_1,j_2} = e_{i_1} \otimes e_{j_1} \otimes e_{i_2} \otimes e_{j_2} - e_{i_1} \otimes e_{j_2} \otimes e_{i_2} \otimes e_{j_1}. \]
Now $w$ belongs to the intersection of two subspaces $\cA \cap \cB$ where $\cA$ is the set of vectors expressible in the form~\eqref{eq:delta-1} and $\cB$ is the set of solutions to~\eqref{eq:delta-2}. A key step in the algorithm of~\cite{JLV} is to compute this intersection of subspaces (which in their notation is $S^2(\cY) \cap I^\perp_2$). The hope is to have no ``spurious'' vectors in this subspace, that is,
\begin{equation}\label{eq:JLV-intersect}
\cA \cap \cB = \Span\{d^{(\ell)} \otimes b^{(\ell)} \otimes d^{(\ell)} \otimes b^{(\ell)} \,:\ \ell \in [r]\}.
\end{equation}
The containment ``$\supseteq$'' in~\eqref{eq:JLV-intersect} always holds, so failure of~\eqref{eq:JLV-intersect} is equivalent to the existence of coefficients $\delta_{\ell_1,\ell_2}$, not all zero, but with $\delta_{\ell,\ell} = 0$, such that the vector~\eqref{eq:delta-1} satisfies~\eqref{eq:delta-2}:
\[ \sum_{1 \le \ell_1 < \ell_2 \le r} \delta_{\ell_1,\ell_2} \langle d^{(\ell_1)} \otimes b^{(\ell_1)} \otimes d^{(\ell_2)} \otimes b^{(\ell_2)} + d^{(\ell_2)} \otimes b^{(\ell_2)} \otimes d^{(\ell_1)} \otimes b^{(\ell_1)}, E_{i_1,i_2;j_1,j_2} \rangle = 0 \]
for $1 \le i_1 < i_2 \le p+1 \text{ and } 1 \le j_1 < j_2 \le n_2$. This reduces to
\[ \sum_{1 \le \ell_1 < \ell_2 \le r} \delta_{\ell_1,\ell_2} \left(d^{(\ell_1)}_{i_1} b^{(\ell_1)}_{j_1} d^{(\ell_2)}_{i_2} b^{(\ell_2)}_{j_2} + d^{(\ell_2)}_{i_1} b^{(\ell_2)}_{j_1} d^{(\ell_1)}_{i_2} b^{(\ell_1)}_{j_2} - d^{(\ell_1)}_{i_1} b^{(\ell_1)}_{j_2} d^{(\ell_2)}_{i_2} b^{(\ell_2)}_{j_1} - d^{(\ell_2)}_{i_1} b^{(\ell_2)}_{j_2} d^{(\ell_1)}_{i_2} b^{(\ell_1)}_{j_1}\right) = 0, \]
which can be written in matrix form as $P \delta = 0$ with $P$ defined in~\eqref{eq:def-P}. Since $P$ has full column rank by assumption~\ref{item:P}, there is no nonzero solution to $P \delta = 0$, so~\eqref{eq:JLV-intersect} holds.

We have now shown that for any $z \in \cX \cap \cY$, the corresponding $w := z \otimes z$ lies in the subspace~\eqref{eq:JLV-intersect}. Our goal is to show that $z$ is a scalar multiple of some $d^{(\ell)} \otimes b^{(\ell)}$. Since the vectors $d^{(\ell)} \otimes b^{(\ell)}$ are linearly independent by assumption~\ref{item:db-li}, Lemma~\ref{lem:unique-rank-1} below implies that the vectors $d^{(\ell)} \otimes b^{(\ell)} \otimes d^{(\ell)} \otimes b^{(\ell)}$ are the only vectors of the form $y \otimes y$ in their span, up to scalar multiple. This means $w$ is a scalar multiple of some $d^{(\ell)} \otimes b^{(\ell)} \otimes d^{(\ell)} \otimes b^{(\ell)}$, and thus $z$ is a scalar multiple of some $d^{(\ell)} \otimes b^{(\ell)}$, as desired.

\begin{lemma}\label{lem:unique-rank-1}
If $z^{(1)},\ldots,z^{(r)} \in \RR^N$ are linearly independent vectors then $z^{(1)} \otimes z^{(1)}, \ldots, z^{(r)} \otimes r^{(r)}$ are the only rank-1 tensors in their span, $\Span\{z^{(\ell)} \otimes z^{(\ell)} \,:\, \ell \in [r]\}$, up to scalar multiple, and furthermore there is an algorithm based on simultaneous diagonalization to recover (scalar multiples of) the vectors $z^{(\ell)} \otimes z^{(\ell)}$ given their span.
\end{lemma}

\begin{proof}
First we give a simple self-contained proof of the uniqueness claim. View each $z^{(\ell)} \otimes z^{(\ell)}$ as a rank-1 matrix, and apply a change of basis (to both rows and columns) so that $z^{(1)},\ldots,z^{(r)}$ become standard basis vectors $e^{(1)},\ldots,e^{(r)}$. Now each $z^{(\ell)} \otimes z^{(\ell)}$ becomes a diagonal matrix with a single nonzero entry at position $(\ell,\ell)$. Any linear combination $\sum_{\ell=1}^r \alpha_\ell \cdot z^{(\ell)} \otimes z^{(\ell)}$ is diagonal and has rank equal to the number of nonzero coefficients $\alpha_\ell$ (and in particular, the vectors $z^{(\ell)} \otimes z^{(\ell)}$ are linearly independent). Therefore the only rank-1 matrices in this span are scalar multiples of a single $z^{(\ell)} \otimes z^{(\ell)}$.

There is also an algorithmic proof using simultaneous diagonalization, which appears in~\cite{JLV}. Given an arbitrary basis $p^{(1)},\ldots,p^{(r)}$ for $\Span\{z^{(\ell)} \otimes z^{(\ell)} \,:\, \ell \in [r]\}$, construct the third-order tensor $R := \sum_{\ell=1}^r e^{(\ell)} \otimes p^{(\ell)} \in \RR^r \otimes \RR^N \otimes \RR^N$ where $e^{(\ell)}$ denotes the $\ell$th standard basis vector. There are coefficients $\beta_{\ell m }$ such that $p^{(\ell)} = \sum_{m=1}^r \beta_{\ell m} \cdot z^{(m)} \otimes z^{(m)}$ for each $\ell \in [r]$, where $\beta = (\beta_{\ell m} \,:\, \ell,m \in [r])$ is an invertible matrix. This means $R$ admits the decomposition
\begin{equation}\label{eq:R-decomp}
R = \sum_{\ell=1}^r e^{(\ell)} \otimes \left(\sum_{m=1}^r \beta_{\ell m} \cdot z^{(m)} \otimes z^{(m)}\right) = \sum_{m=1}^r \left(\sum_{\ell=1}^r \beta_{\ell m} \cdot e^{(\ell)}\right) \otimes z^{(m)} \otimes z^{(m)}.
\end{equation}
The vectors $\{z^{(m)} \,:\, m \in [r]\}$ are linearly independent by assumption, and the vectors $\{\sum_{\ell=1}^r \beta_{\ell m} \cdot e^{(\ell)} \,:\, m \in [r]\}$ are linearly independent because they are the columns of the invertible matrix $\beta$. This verifies the conditions for simultaneous diagonalization (see~\cite[Theorem~3.1.3]{moitra-book}) so the decomposition on the right-hand side of~\eqref{eq:R-decomp} is unique and can be recovered by an efficient algorithm. This allows recovery of the vectors $z^{(\ell)} \otimes z^{(\ell)}$ (up to re-ordering and scalar multiple).
\end{proof}

Finally, the algorithmic claim of Lemma~\ref{lem:JLV} also follows from the above, since the~\cite{JLV} algorithm first computes the intersection of subspaces $\cA \cap \cB$, which is equal to~\eqref{eq:JLV-intersect}, and then applies simultaneous diagonalization as in Lemma~\ref{lem:unique-rank-1} to recover the vectors $d^{(\ell)} \otimes b^{(\ell)}$ from this subspace.

\subsection{Proof of Theorem~\ref{thm:generic-success}}
\label{sec:pf-generic-success}

First we verify the claim $\bar{p} \ge 2$. To show $p \ge 1$, it suffices by the choice of $p$ to have $q n_3 / (n_2+n_3) \ge 1$. This follows by combining two facts: $q > 1+\alpha$, which is a consequence of the assumption on $q$, and $n_3/(n_2+n_3) \ge 1/(1+\alpha)$, which is a consequence of the definition of $\alpha$. Similarly, to show $q-p \ge 2$, it suffices by the choice of $p$ to have $q n_3/(n_2+n_3) < q-1$, i.e., $q n_2/(n_2+n_3) > 1$, which again follows from $q > 1+\alpha$ and $n_2/(n_2+n_3) \ge 1/(1+\alpha)$.

Now conditions~\ref{item:a1} and~\ref{item:d-distinct} are clearly met for generic components, since $\bar{p} \ge 2$. The following lemma establishes~\ref{item:db-li}--\ref{item:bc-li}, provided $r \le \min\{(p+1)n_2,(q-p)n_3,n_2 n_3\}$.

\begin{lemma}
Suppose $x^{(1)} \otimes y^{(1)}, \ldots, x^{(r)} \otimes y^{(r)} \in \RR^m \otimes \RR^n$ with the components $x^{(\ell)}, y^{(\ell)}$ generically chosen. If $r \le mn$ then $\{x^{(\ell)} \otimes y^{(\ell)} \,:\, \ell \in [r]\}$ are linearly independent.
\end{lemma}
\begin{proof}
Form the $mn \times r$ matrix whose columns are $x^{(\ell)} \otimes y^{(\ell)}$. It suffices to identify a square submatrix using all the columns whose determinant is a nonzero polynomial in the entries of $x^{(\ell)}, y^{(\ell)}$. For this, it suffices to demonstrate a choice of the vectors $x^{(\ell)}, y^{(\ell)}$ for which $\{x^{(\ell)} \otimes y^{(\ell)} \,:\, \ell \in [r]\}$ are linearly independent. Fix an injection $[r] \to [m] \times [n]$. For each $\ell \in [r]$, if $\ell \mapsto (i,j)$ then set $x^{(\ell)} = e^{(i)}$ and $y^{(\ell)} = e^{(j)}$ where $e^{(i)}$ is the $i$th standard unit basis vector.
\end{proof}

Condition~\ref{item:M} on $M$ follows from the proof of Theorem~\ref{thm:rank-det}, namely~\eqref{eq:before-sub-p}, provided
\[ r \le q\left(\min\left\{\frac{n_2}{q-p}, \, \frac{n_3}{p+1}\right\} - 1\right). \]
This condition is unchanged under the operation that swaps $n_2$ with $n_3$ and replaces $p$ with $q-p-1$. Therefore, condition~\ref{item:M-prime} on $M'$ also holds under the same condition.

The following lemma establishes condition~\ref{item:N} on $N$.

\begin{lemma}\label{lem:N}
In the setting of Theorem~\ref{thm:uniqueness} with parameters $q \le n_1$ and $\bar{p} := \min\{p+1,q-p\} \ge 2$, if the components $a^{(\ell)}, b^{(\ell)}, c^{(\ell)}$ are generically chosen and
\[ r \le \frac{qn_2}{q-p}\left(1 - \frac{2}{q}\right) - \frac{q^3}{4}, \]
then $N$ has full column rank.
\end{lemma}

\noindent The proof is deferred to Section~\ref{sec:pf-N}. By symmetry, condition~\ref{item:N-prime} on $N'$ also holds under the analogous condition
\[ r \le \frac{qn_3}{p+1}\left(1 - \frac{2}{q}\right) - \frac{q^3}{4}. \]

For condition~\ref{item:P} on $P$ we appeal to the results of~\cite{JLV}. As explained in Section~\ref{sec:pf-JLV}, condition~\ref{item:P} is equivalent to~\eqref{eq:JLV-intersect}, which is established by~\cite[Corollary~3]{JLV} for generic components under the condition $r \le p(n_2-1)/4$. Similarly, condition~\ref{item:P-prime} on $P'$ holds provided $r \le (q-p-1)(n_3-1)/4$.

To summarize, we have the conditions
\begin{enumerate}[label=(\alph*)]
\item\label{item:cond-1} $r \le \min\{(p+1)n_2, \, (q-p)n_3\}$,
\item\label{item:cond-2} $r \le n_2 n_3$,
\item\label{item:cond-M} $r \le \min\left\{\frac{qn_2}{q-p}, \, \frac{qn_3}{p+1}\right\} - q$,
\item\label{item:cond-N} $r \le \min\left\{\frac{qn_2}{q-p}, \, \frac{qn_3}{p+1}\right\}\left(1-\frac{2}{q}\right) - \frac{q^3}{4}$,
\item\label{item:cond-JLV} $r \le \frac{1}{4}\min\{p(n_2-1), \, (q-p-1)(n_3-1)\}$.
\end{enumerate}
Condition~\ref{item:cond-1} is subsumed by~\ref{item:cond-JLV} because, e.g., $p(n_2-1)/4 \le pn_2 \le (p+1)n_2$. The assumption $\bar{p} \ge 2$ implies $p \ge 1$ and $q \ge p+2 \ge 3$, so $q^3/4 > q$. This means condition~\ref{item:cond-M} is subsumed by~\ref{item:cond-N}. Also,~\ref{item:cond-JLV} implies $n_2,n_3 \ge 2$ (or else $r \le 0$), so~\ref{item:cond-N} implies
\[ r \le \min\left\{\frac{qn_2}{q-p}, \, \frac{qn_3}{p}\right\} \le \max_{x \in (0,1)} \min\left\{\frac{n_2}{1-x}, \, \frac{n_3}{x}\right\} = n_2 + n_3 \le 2 \max\{n_2,n_3\} \le n_2 n_3, \]
so~\ref{item:cond-2} is also subsumed. This leaves only~\ref{item:cond-N} and~\ref{item:cond-JLV}.

Now using the choice of $p = \lfloor q \cdot n_3/(n_2+n_3) \rfloor$ and the bounds~\eqref{eq:p-bound-1},\eqref{eq:p-bound-2} we have
\[ \min\left\{\frac{qn_2}{q-p}, \, \frac{qn_3}{p+1}\right\}\left(1-\frac{2}{q}\right) \ge (n_2+n_3)\left(1-\frac{1+\alpha}{q}\right)\left(1-\frac{2}{q}\right) \ge (n_2+n_3)\left(1-\frac{3+\alpha}{q}\right), \]
so condition~\ref{item:cond-N} can be replaced with the sufficient condition
\begin{equation}\label{eq:r-main-condition}
r \le (n_2+n_3)\left(1-\frac{3+\alpha}{q}\right) - \frac{q^3}{4}.
\end{equation}

Finally, we will show that the previous condition subsumes~\ref{item:cond-JLV}, given the assumption on $q$. Using the choice of $p$,
\[ \frac{1}{4} p (n_2-1) \ge \frac{1}{4} \left(q \cdot \frac{n_3}{n_2+n_3} - 1\right)(n_2-1) = \frac{(q-1)n_3 - n_2}{4(n_2+n_3)}(n_2-1) \ge \frac{(q-1-\alpha)n_3}{4(n_2+n_3)}(n_2-1) \]
and similarly,
\[ \frac{1}{4}(q-p-1)(n_3-1) \ge \frac{1}{4}\left(q-q \cdot \frac{n_3}{n_2+n_3} - 1\right)(n_3-1) \ge \frac{(q-1-\alpha)n_2}{4(n_2+n_3)}(n_3-1). \]
Therefore, the right-hand side of condition~\ref{item:cond-JLV} can be replaced by
\begin{align*}
\frac{q-1-\alpha}{4(n_2+n_3)}(n_2-1)(n_3-1) &\ge \frac{q-1-\alpha}{4} \left(\frac{n_2 n_3}{n_2+n_3} - 1\right) = \frac{q-1-\alpha}{4} \left[\frac{\alpha}{(1+\alpha)^2}(n_2+n_3) - 1\right]
\intertext{where the last step follows by checking the cases $n_3 = \alpha n_2$ and $n_2 = \alpha n_3$}
&\ge \frac{(q-1-\alpha)\alpha}{4(1+\alpha)^2} (n_2+n_3) - \frac{q}{4} \ge n_2 + n_3 - \frac{q}{4}
\end{align*}
using the assumption $q \ge (4+5\alpha)(1+1/\alpha)$, so condition~\ref{item:cond-JLV} is subsumed by~\eqref{eq:r-main-condition}.

\subsection{Proof of Lemma~\ref{lem:N}}
\label{sec:pf-N}

It will be convenient to work not with $N$ but with an equivalent matrix, $L$. Recall the derivation of $N$ from Section~\ref{sec:pf-subspace-intersect}. There, we used the columns $\{U \,:\, 1 \in U\}$ as a basis for the column span of $A(a^{(\ell)})$. By Lemma~\ref{lem:rank-A}, and assuming $a^{(\ell)}_i \ne 0$ due to generic components, we can just as easily use the basis of columns $\{U \,:\, i \in U\}$ for any $i \in V := [p+1]$. To define $L$, we will use a different choice of $i \in [p+1]$ for each $\ell$. The purpose of this is to create symmetry among the elements $i \in [p+1]$, so that element $1$ no longer plays a distinguished role. (However, there remains a distinction between $i \in [p+1]$ versus $i \in [q] \setminus [p+1]$, since we have fixed $V = [p+1]$.)

Formally, choose the smallest integer $\bar{r}$ for which $r \le \bar{r}q(p+1)(q-p-1)$. It suffices to verify the extreme case $r = \bar{r}q(p+1)(q-p-1)$ (i.e., we will show that even when $r$ is increased to add additional columns, all the columns are linearly independent). We will index $[r]$ by tuples $(i,j,\ell)$ with $i \in [p+1]$, $j \in [q] \setminus [p+1]$, and $\ell \in [\bar{r}q]$; this way, each element of $[r]$ has an assigned $(i,j)$ pair, where $i$ plays a role in choosing the basis described above, and the role of $j$ will appear later. Define the matrix $L$ with rows indexed by
\[ \{(m,S) \,:\, m \in [n_2], \, S \subseteq [q], \, |S| = p, \, S \not\subseteq [p+1]\} \]
and columns indexed by
\begin{equation}\label{eq:index-cols}
\{(i,j,\ell,U) \,:\, i \in [p+1], \, j \in [q] \setminus [p+1], \, \ell \in [\bar{r}q], \, U \subseteq [q], \, |U| = p+1, \, i \in U, \, U \neq [p+1] \},
\end{equation}
with entries
\[ L_{mS,ij \ell U} = b^{(i,j,\ell)}_m \cdot \sum_{k=1}^q \One_{U = S \sqcup \{k\}} \cdot \sigma(U,k) \cdot a^{(i,j,\ell)}_k, \]
where we have indexed $[r]$ by tuples $(i,j,\ell)$ as described above. Our goal is to show that $L$ has full column rank, which implies the same for $N$, since both these conditions are equivalent to~\eqref{eq:intersect-colspan} via the argument in Section~\ref{sec:pf-subspace-intersect} (along with the discussion above).

It suffices to identify a nonzero minor of $L$ that uses all the columns. To this end, we will describe a pairing that assigns every column $(i,j,\ell,U)$ to some row $(m,S)$ in such a way that each row is used 0 or 1 times. Imagine the columns grouped into $\bar{r}q$ ``epochs,'' one for each $\ell$ value, and each epoch is further subdivided into $(p+1)(q-p-1)$ ``pillars,'' one for each $(i,j)$ pair. The ordering of the epochs is unimportant, but let's fix one arbitrary ordering, and similarly for the ordering of pillars within an epoch and the ordering of columns within a pillar. Call the first $\bar{r}(p+1)$ epochs ``type 1'' and the remaining $\bar{r}(q-p-1)$ epochs ``type 2.'' We first describe which $S$ value gets assigned to a given column $(i,j,\ell,U)$:
\begin{itemize}
    \item For a type 1 epoch (i.e., $1 \le \ell \le \bar{r}(p+1)$), remove $i$ from $U$ to produce $S$: let $S = U \setminus \{i\}$. Note that $i \in U$ is guaranteed by~\eqref{eq:index-cols}.
    \item For a type 2 epoch (i.e., $\bar{r}(p+1) < \ell \le \bar{r}q$), remove $j$ if possible, and otherwise remove $i$. That is, if $j \in U$ and $|U \cap [p+1]| < p$ then let $S = U \setminus \{j\}$, and otherwise let $S = U \setminus \{i\}$. (Recall that $S \subseteq [p+1]$ is forbidden, so $j$ cannot be removed if $|U \cap [p+1]| = p$.)
\end{itemize}
Now that the $S$ values are assigned, the $m$ values are assigned based on the ``first available'' rule. That is, column $(i,j,\ell,U)$ gets assigned to row $(m,S)$ where $S$ is specified above and $m$ is the minimum $m'$ for which $(m',S)$ has not already been assigned to some previous column. Consider the submatrix $L'$ of $L$ formed by all columns of $L$ and only those rows that were assigned to some column by the above procedure. Our goal is to show that $\det(L')$ is a nonzero polynomial in the variables $b_m^{(i,j,\ell)}$ and $a_k^{(i,j,\ell)}$.

Expand the determinant $\det(L')$ as a sum over permutations. Each term in the sum gives a single monomial (multiplied by a coefficient $\pm 1$ or $0$). Consider the term in the sum corresponding to the pairing of rows and columns described above. We claim that this term produces a nonzero monomial that is unique in the sense that no other term produces a (nonzero) scalar multiple of this monomial. Once we establish this claim, we are done. To prove the claim we need to argue that, given our monomial of interest, it is possible to deduce the unique pairing of rows and columns that produced it. Within a fixed pillar $(i,j,\ell)$, the powers of $a^{(i,j,\ell)}_k$ in the monomial reveal which multiset of $k$ values are removed from the $U$'s in this pillar to form the associated $S$'s. By design, there are only 1 or 2 distinct $k$ values, namely $\{i\}$ or $\{i,j\}$ (for a type 1 or type 2 epoch, respectively). In either case, the specific pairing between $U$'s and $S$'s can be deduced (using the fact that a type 2 epoch removes $j$ whenever possible). The $S$'s appearing within a pillar are all distinct (when $j$ is removed, the resulting $S$ contains $i$, but when $i$ is removed, it does not). Now the powers of $b_m^{(i,j,\ell)}$ in the monomial reveal which multiset of $m$ values are used in each pillar. Recalling the ``first available'' rule, this allows the $m$ value for each column to be deduced. This completes the proof that $\det(L')$ is a nonzero polynomial.

Finally we need to determine what value of $n_2$ is required for the above construction. The minimum allowable value for $n_2$ is the largest $m$ value used in the construction. For each $S$, we need to count how many columns $(i,j,\ell,U)$ get assigned to $S$.

We first consider the special case where $|S \cap [p+1]|$ takes its maximum possible value, $p-1$. This case needs separate consideration due to the special rule for type 2 epochs: if $|U \cap [p+1]| = p$ then $j$ cannot be removed. Fix $S$ with $|S \cap [p+1]| = p-1$. In a type 1 epoch $\ell$, the number of tuples $(i,j,U)$ assigned to $S$ is $2(q-p-1)$ because there are 2 choices for $i \notin S$, then $q-p-1$ choices for $j$, and then $U = S \sqcup \{i\}$. In a type 2 epoch, we have two different contributions to count. First consider $U$ with $|U \cap [p+1]| = p$, and recall that these are never allowed to remove $j$. This gives $2(q-p-1)$ tuples $(i,j,U)$, since there are again $2$ choices for $i \notin S$, then $q-p-1$ choices for $j$, and $U = S \sqcup \{i\}$. We also need to consider $U$ with $|U \cap [p+1]| = p-1$ (i.e., cases where $j$ is removed). This gives $(p-1)(q-p-2)$ tuples $(i,j,U)$, since there are $p-1$ choices for $i \in S$, then $q-p-2$ choices for $j \notin S$, and $U = S \sqcup \{j\}$ (we need $i \in S$ so that $i \in U$, as required by~\eqref{eq:index-cols}). In total, the number of times that $S$ is used is
\begin{align*}
&\bar{r}(p+1) \cdot 2(q-p-1) + \bar{r}(q-p-1) \cdot [2(q-p-1) + (p-1)(q-p-2)] \\
&\qquad = \bar{r}(q-p-1)[2q + (p-1)(q-p-2)] \\
&\qquad = \bar{r}(p+1)(q-p-1)\left(q-p + \frac{2}{p+1}\right).
\end{align*}

Now consider the remaining case: fix $S$ with $h := |S \cap [p+1]| < p-1$. In a type 1 epoch $\ell$, the number of tuples $(i,j,U)$ assigned to $S$ is $(p+1-h)(q-p-1)$ because there are $p+1-h$ choices for $i \notin S$, then $q-p-1$ choices for $j$, and $U = S \sqcup \{i\}$. In a type 2 epoch, we again have two contributions to count. First consider $U$ with $|U \cap [p+1]| = h+1$ (i.e., $i$ is removed). This gives $(p+1-h)(q-2p+h-1)$ tuples $(i,j,U)$, since there are $p+1-h$ choices for $i \notin S$, then $q-2p+h-1$ choices for $j \notin S$, and $U = S \sqcup \{i\}$. Now consider $U$ with $|U \cap [p+1]| = h$ (i.e., $j$ is removed). This gives $h(q-2p+h-1)$ tuples $(i,j,U)$, since there are $h$ choices for $i \in S$, then $q-2p+h-1$ choices for $j \notin S$, and $U = S \sqcup \{j\}$. In total, the number of times that $S$ is used is
\begin{align*}
&\bar{r}(p+1) \cdot (p+1-h)(q-p-1) + \bar{r}(q-p-1) \cdot [(p+1-h)(q-2p+h-1) + h(q-2p+h-1)] \\
&\qquad = \bar{r}(q-p-1)[(p+1)(p+1-h) + (p+1)(q-2p+h-1)] \\
&\qquad = \bar{r}(p+1)(q-p-1)(q-p),
\end{align*}
which is dominated by the first case.

To recap, we need an integer $\bar{r}$ such that $r \le \bar{r}q(p+1)(q-p-1)$ and $n_2 \ge \bar{r}(p+1)(q-p-1)(q-p + 2/(p+1))$, or equivalently,
\[ \frac{r}{q(p+1)(q-p-1)} \le \bar{r} \le \frac{n_2}{(p+1)(q-p-1)(q-p+2/(p+1))}. \]
It suffices to have
\[ \frac{r}{q(p+1)(q-p-1)} \le \frac{n_2}{(p+1)(q-p-1)(q-p+2/(p+1))} - 1, \]
that is,
\[ r \le \frac{qn_2}{q-p+2/(p+1)} - q(p+1)(q-p-1). \]
We have the bounds $(p+1)(q-p-1) \le \max_{x \in \RR} x(q-x) = q^2/4$, and
\[ \frac{qn_2}{q-p+2/(p+1)} \ge \frac{qn_2}{q-p}\left(1 - \frac{2}{(p+1)(q-p)}\right) \ge \frac{qn_2}{q-p}\left(1 - \frac{2}{q}\right). \]
Thus it suffices to have
\[ r \le \frac{qn_2}{q-p}\left(1 - \frac{2}{q}\right) - \frac{q^3}{4} \]
as claimed.

\section{Proofs for Lower Bounds}

\subsection{Matrix Decomposition Lemma}

The following result of~\cite{barriers-rank} is a core component of our lower bounds. (In~\cite{barriers-rank}, only square matrices are considered, but the same proof applies for rectangular matrices.) Recall the notion of symbolic rank from Section~\ref{sec:genericity}.

\begin{lemma}[\cite{barriers-rank}, Lemma~3.2]
\label{lem:matrix-decomp}
Let $M(x) \in \RR[x]^{m_1 \times m_2}$ be an $m_1 \times m_2$ matrix whose entries are homogeneous degree-$d$ polynomials in the variables $x_1,\ldots,x_p$. If $M(x)$ has symbolic rank $r$ then there are vectors $f_1(x),\ldots,f_r(x) \in \RR[x]^{m_1}$ and $g_1(x),\ldots,g_r(x) \in \RR[x]^{m_2}$ such that
\[ M(x) = \sum_{t=1}^r H_d[f_t(x) g_t(x)^\top], \]
where $H_d(\cdot)$ denotes the degree-$d$ homogeneous part (applied entry-wise).
\end{lemma}

\subsection{Proof of Theorem~\ref{thm:lower-ky}}
\label{sec:pf-lower-ky}

Let $s$ be the symbolic rank of $A(a)$, which is also the generic rank of $M(a^{(1)} \otimes b^{(1)} \otimes c^{(1)})$. The entries of $A(a)$ are degree-$1$ homogeneous polynomials in $a$, so by Lemma~\ref{lem:matrix-decomp} there are vectors of polynomials $f_1(a),\ldots,f_s(a)$ and $g_1(a),\ldots,g_s(a)$ such that
\[ A(a) = \sum_{t=1}^s H_1[f_t(a)g_t(a)^\top]. \]
Note that
\[ H_1[f_t(a)g_t(a)^\top] = H_0[f_t] H_1[g_t(a)]^\top + H_1[f_t(a)] H_0[g_t]^\top, \]
where we write, e.g., $H_0[f_t]$ instead of $H_0[f_t(a)]$ to emphasize that this is a constant (degree-$0$) vector. Now we can write
\[ M(a \otimes b \otimes c) = A(a) \otimes (bc^\top) = \sum_{t=1}^s \{(H_0[f_t] \otimes b)(H_1[g_t(a)] \otimes c)^\top + (H_1[f_t(a)] \otimes b)(H_0[g_t] \otimes c)^\top\} \]
and so, since $T \mapsto M(T)$ is linear,
\[ M\left(\sum_{\ell=1}^r a^{(\ell)} \otimes b^{(\ell)} \otimes c^{(\ell)}\right) = M_0 + M_1 \]
where
\[ M_0 := \sum_{t=1}^s \sum_{\ell=1}^r (H_0[f_t] \otimes b^{(\ell)})(H_1[g_t(a^{(\ell)})] \otimes c^{(\ell)})^\top \]
and
\[ M_1 := \sum_{t=1}^s \sum_{\ell=1}^r (H_1[f_t(a^{(\ell)})] \otimes b^{(\ell)})(H_0[g_t] \otimes c^{(\ell)})^\top. \]
Note that
\[ \Colspan(M_0) \subseteq \Span\{H_0[f_t] \otimes \RR^{n_2} \,:\, t \in [s]\} \]
and
\[ \Rowspan(M_1) \subseteq \Span\{H_0[g_t] \otimes \RR^{n_3} \,:\, t \in [s]\}, \]
so for any setting of the variables $a^{(\ell)},b^{(\ell)},c^{(\ell)}$ we have
\[ \rank\left(M\left(\sum_{\ell=1}^r a^{(\ell)} \otimes b^{(\ell)} \otimes c^{(\ell)}\right)\right) \le \rank(M_0) + \rank(M_1) \le sn_2 + sn_3. \]
Additivity would imply that this rank must also equal $rs$, from which we conclude $r \le n_2 + n_3$.

\subsection{Proof of Theorem~\ref{thm:lower-linear}}
\label{sec:pf-lower-linear}

Let $s$ be the generic rank of $M(a^{(1)} \otimes b^{(1)} \otimes c^{(1)})$, which is also the symbolic rank of the matrix $M(a \otimes b \otimes c)$, viewed as a matrix of polynomials in the variables $x = (a_1,\ldots,a_{n_1}, b_1,\ldots,b_{n_2}, c_1,\ldots,c_{n_3})$. The entries of $M(a \otimes b \otimes c)$ are degree-$3$ homogeneous polynomials, so by Lemma~\ref{lem:matrix-decomp} there are vectors of polynomials $f_1(x),\ldots,f_s(x)$ and $g_1(x),\ldots,g_s(x)$ such that
\[ M(a \otimes b \otimes c) = \sum_{t=1}^s H_3[f_t(x)g_t(x)^\top]. \]
The variables $x$ are naturally partitioned into $X_1 := \{a_i \,:\, i \in [n_1]\}$, $X_2 := \{b_i \,:\, i \in [n_2]\}$, and $X_3 := \{c_i \,:\, i \in [n_3]\}$. For $S \subseteq [3]$, let $H_S(\cdot)$ denote the terms that are degree $1$ in $X_i$ for each $i \in S$ and degree $0$ in $X_i$ for each $i \notin S$. Note that all terms in $M(a \otimes b \otimes c)$ are degree $1$ in each of $X_1,X_2,X_3$, so we can rewrite the above as
\[ M(a \otimes b \otimes c) = \sum_{t=1}^s H_{\{1,2,3\}}[f_t(x)g_t(x)^\top]. \]
We can write
\[ H_{\{1,2,3\}}[f_t(x)g_t(x)^\top] = \sum_{S \subseteq [3]} H_S[f_t(x)]H_{[3] \setminus S}[g_t(x)]^\top, \]
and so
\[ M\left(\sum_{\ell=1}^r a^{(\ell)} \otimes b^{(\ell)} \otimes c^{(\ell)}\right) = \sum_{S \subseteq [3]} M_S \]
where
\[ M_S := \sum_{t=1}^s \sum_{\ell=1}^r H_S[f_t(x^{(\ell)})]H_{[3] \setminus S}[g_t(x^{(\ell)})]^\top. \]
Now, for any setting of the variables $x^{(\ell)}$, we bound the rank of each $M_S$. First, $\rank(M_\emptyset) \le s$ because $\Colspan(M_\emptyset) \subseteq \Span\{H_\emptyset[f_t] \,:\, t \in [s]\}$, noting that $H_\emptyset[f_t(x)]$ is a constant (degree-$0$) vector, which we denote $H_\emptyset[f_t]$. Next, note that $H_{\{1\}}[f_t(x)]$ has only ``$a$'' variables (not ``$b$'' or ``$c$''), which we denote by $H_{\{1\}}[f_t(x^{(\ell)})] = H_{\{1\}}[f_t(a=a^{(\ell)})]$. Each $H_{\{1\}}[f_t(a=a^{(\ell)})]$ lies in $\Span\{H_{\{1\}}[f_t(a=e^{(j)})] \,:\, j \in [n_1]\}$ where $e^{(j)}$ are unit basis vectors. This means $\rank(M_{\{1\}}) \le sn_1$ because $\Colspan(M_{\{1\}}) \subseteq \Span\{H_{\{1\}}[f_t(a=e^{(j)})] \,:\, t \in [s], \, j \in [n_1]\}$. Similarly, $\rank(M_{\{2\}}) \le sn_2$ and $\rank(M_{\{3\}}) \le sn_3$. Applying the same argument to the row span, we also have $\rank(M_{\{1,2,3\}}) \le s$, $\rank(M_{\{2,3\}}) \le sn_1$, $\rank(M_{\{1,3\}}) \le sn_2$, and $\rank(M_{\{1,2\}}) \le sn_3$. Putting it together,
\[ \rank\left(M\left(\sum_{\ell=1}^r a^{(\ell)} \otimes b^{(\ell)} \otimes c^{(\ell)}\right)\right) \le \sum_{S \subseteq [3]} \rank(M_S) \le 2s(n_1+n_2+n_3+1). \]
Additivity would imply that this rank is $rs$, from which we conclude $r \le 2(n_1+n_2+n_3+1)$.

\begin{remark}\label{rem:improve-8n}
The proof above is similar to that of Theorem~4.4 of~\cite{barriers-rank}, where they give a bound of $r \le 8n$ for $n \times n \times n$ tensors. We now explain the source of our improvement from $8n$ to $6n+2$. When decomposing $M$ into 8 terms $M_S$ for $S \subseteq [3]$, the original proof of~\cite{barriers-rank} bounds the rank of each term by $sn$. We instead keep track of the specific structure of each term and see that two of them, namely $M_\emptyset$ and $M_{\{1,2,3\}}$, admit a better bound of $s$.
\end{remark}

\subsection{Proof of Theorem~\ref{thm:lower-d}}
\label{sec:pf-lower-d}

Let $s$ be the symbolic rank of
\begin{equation}\label{eq:N-plugin}
\tilde{M}(b^{(1,1)} \otimes \cdots \otimes b^{(k,1)},\ldots,b^{(1,d)} \otimes \cdots \otimes b^{(k,d)})
\end{equation}
in the variables $b = \{b^{(i,j)} \,:\, i \in [k], \, j \in [d]\}$; this is also the generic rank of $\tilde{M}(T^{(1)},\ldots,T^{(d)})$. The entries of~\eqref{eq:N-plugin} are degree-$kd$ homogeneous polynomials in the $b$ variables, so by Lemma~\ref{lem:matrix-decomp},~\eqref{eq:N-plugin} admits the decomposition
\[ \sum_{t=1}^s H_{kd}[f_t(b)g_t(b)^\top]. \]
The variables $b$ are naturally partitioned into $\bigsqcup_{i \in [k], j \in [d]} B_{ij}$ where $B_{ij} := \{b_m^{(i,j)} \,:\, m \in [n_i]\}$. For $S \subseteq [k] \times [d]$, let $H_S(\cdot)$ denote the terms that are degree $1$ in $B_{ij}$ for each $(i,j) \in S$ and degree $0$ in $B_{ij}$ for each $(i,j) \notin S$. Note that all terms in~\eqref{eq:N-plugin} are degree $1$ in every $B_{ij}$, so we can rewrite the above decomposition as
\[ \sum_{t=1}^s H_{[k] \times [d]}[f_t(b)g_t(b)^\top]. \]
With $\bar{S} := ([k] \times [d]) \setminus S$, we can write
\[ H_{[k] \times [d]}[f_t(b)g_t(b)^\top] = \sum_{S \subseteq [k] \times [d]} H_S[f_t(b)]H_{\bar{S}}[g_t(b)]^\top, \]
and so
\[ \sum_{1 \le \ell_1 < \cdots < \ell_d \le r} \tilde{M}(T^{(\ell_1)},\ldots,T^{(\ell_d)}) = \sum_{S \subseteq [k] \times [d]} M_S \]
where
\[ M_S := \sum_{t=1}^s \sum_{1 \le \ell_1 < \cdots < \ell_d \le r} H_S[f_t(b^{(i,j)} = a^{(i,\ell_j)})]H_{\bar{S}}[g_t(b^{(i,j)} = a^{(i,\ell_j)})]^\top. \]
Here, $H_S[f_t(b^{(i,j)} = a^{(i,\ell_j)})]$ is obtained from $H_S[f_t(b)]$ by making the substitution $b^{(i,j)} = a^{(i,\ell_j)}$ for each $(i,j) \in S$; note that the variables $b^{(i,j)}$ with $(i,j) \notin S$ do not appear. Similarly, $H_{\bar{S}}[g_t(b^{(i,j)} = a^{(i,\ell_j)})]$ is obtained from $H_{\bar{S}}[g_t(b)]$ by making the substitution $b^{(i,j)} = a^{(i,\ell_j)}$ for each $(i,j) \notin S$. Note that $H_S[f_t(b)]$ is a multilinear function of the $d$ tensors $\{\bigotimes_{i \,:\, (i,j) \in S} b^{(i,j)} \,:\, j \in [d]\}$. After making the substitution $b^{(i,j)} = a^{(i,\ell_j)}$, we have
\[ \bigotimes_{i \,:\, (i,j) \in S} b^{(i,j)} \subseteq \bigotimes_{i \,:\, (i,j) \in S} \RR^{n_i}, \]
where the right-hand side is a subspace of dimension $\prod_{i \,:\, (i,j) \in S} n_i$. We also have
\[ \bigotimes_{i \,:\, (i,j) \in S} b^{(i,j)} \subseteq \Span\left\{\bigotimes_{i \,:\, (i,j) \in S} a^{(i,\ell)} \,:\, \ell \in [r]\right\}, \]
where the right-hand side is a subspace of dimension $\le r$. For each $S,j$ we will choose the better of these two bounds. We conclude, for any setting of the variables $a^{(i,\ell)}$,
\[ \dim \Colspan(M_S) \le s \prod_{j \in [d]} \min\left\{\prod_{i \,:\, (i,j) \in S} n_i \, , \; r\right\}. \]
By the same argument applied to the row span,
\[ \dim \Rowspan(M_S) \le s \prod_{j \in [d]} \min\left\{\prod_{i \,:\, (i,j) \in \bar{S}} n_i \, , \; r\right\}. \]
For each $S$ we are free to choose the better of these two bounds to bound $\rank(M_S)$. For each $S$, call an index $j \in [d]$ ``good'' if
\[ \prod_{i \,:\, (i,j) \in S} n_i \le n_* := \max_{U \subseteq [k]} \, \min\left\{\prod_{i \in U} n_i, \, \prod_{i \in [k] \setminus U} n_i\right\}. \]
Each $j \in [d]$ must be good for either $S$ or $\bar{S}$. Therefore, either $S$ or $\bar{S}$ has the property that at least $d/2$ indices $j$ are good. If it's $S$ then use the colspan bound for $\rank(M_S)$, and otherwise use the rowspan bound. For ``bad'' indices, use the bound of $r$ in the second term in $\min\{ \cdots\}$. We conclude
\[ \rank(M_S) \le s (n_*)^{d/2} r^{d/2} \]
and so
\[ \rank\left(\sum_{1 \le \ell_1 < \cdots < \ell_d \le r} \tilde{M}(T^{(\ell_1)},\ldots,T^{(\ell_d)})\right) \le \sum_{S \subseteq [k] \times [d]} \rank(M_S) \le 2^{kd} \cdot s (n_*)^{d/2} r^{d/2}. \]
Additivity would imply that this rank must equal $\binom{r}{d} s$, so we have
\[ 2^{kd} (n_*)^{d/2} r^{d/2} \ge \binom{r}{d} \ge \left(\frac{r}{d}\right)^d, \]
which gives $r \le 4^k d^2 \, n_*$ as claimed.

\appendix

\section{Trivial Flattenings}
\label{app:trivial}

The following result shows that trivial flattenings achieve rank detection up to a particular rank, as claimed in Section~\ref{sec:trivial}.

\begin{theorem}
If $T = \sum_{\ell=1}^r a^{(1,\ell)} \otimes \cdots \otimes a^{(k,\ell)}$ is $n_1 \times \cdots \times n_k$ with the components generically chosen, and
\begin{equation}\label{eq:S-max}
r \le n_* := \max_{S \subseteq [k]} \, \min\left\{\prod_{i \in S} n_i, \, \prod_{i \in [k] \setminus S} n_i\right\},
\end{equation}
then the matrix $M^\triv(T;S)$ defined in~\eqref{eq:def-triv}, with $S$ taken to be any maximizer of~\eqref{eq:S-max}, has rank exactly $r$.
\end{theorem}

\begin{proof}
The trivial flattening takes the form
\[ M^\triv(T;S) = \sum_{\ell=1}^r b^{(\ell)} c^{(\ell)\top} \]
where
\[ b^{(\ell)} := \bigotimes_{i \in S} a^{(i,\ell)} \qquad\text{and}\qquad c^{(\ell)} := \bigotimes_{i \in [k] \setminus S} a^{(i,\ell)}. \]
To show $\rank(M^\triv) = r$, it suffices to show the following: $\{b^{(\ell)} \,:\, \ell \in [r]\}$ are linearly independent and $\{c^{(\ell)} \,:\, \ell \in [r]\}$ are linearly independent. To see this: linear independence implies that the matrix $B$ with columns $b^{(\ell)}$ has an $r \times r$ nonsingular submatrix, and similarly for $C$ with rows $c^{(\ell)}$, and the product of these two submatrices appears as an $r \times r$ nonsingular submatrix of $M^\triv$. We focus on the $b^{(\ell)}$ vectors, as the proof for the $c^{(\ell)}$ vectors is identical.

Since rank (of the $B$ matrix) cannot exceed generic rank (see Section~\ref{sec:genericity}), it suffices to exhibit a single choice of the variables $a^{(i,\ell)}$ such that the $b^{(\ell)}$ vectors are linearly independent. Provided $r \le \prod_{i \in S} n_i$, we can choose the variables $a^{(i,\ell)}$ so that the $b^{(\ell)}$ are distinct standard unit basis vectors. Formally, fix an injection $[r] \to \prod_{i \in S} [n_i]$ and write $\ell \mapsto (\phi_i(\ell) \,:\, i \in S)$ where $\phi_i(\ell) \in [n_i]$. Then for $i \in S$, set $a^{(i,\ell)}$ equal to the standard unit basis vector $e^{(\phi_i(\ell))} \in \RR^{n_i}$, so that $b^{(\ell)} = \bigotimes_{i \in S} e^{(\phi_i(\ell))}$. Applying the same argument to the $c^{(\ell)}$ vectors, we also incur the requirement $r \le \prod_{i \in [k] \setminus S} n_i$.
\end{proof}

\section{Commuting Extensions}
\label{app:comm-ext}

For $r \ge n$, a tuple of matrices $(Z_1,\ldots,Z_m)$ in $\RR^{r \times r}$ is called a \emph{commuting extension} of a tuple of matrices $(A_1,\ldots,A_m)$ in $\RR^{n \times n}$ if the $Z_i$ pairwise commute and the upper-left $n \times n$ submatrix of each $Z_i$ is equal to $A_i$. The algorithmic task of finding a commuting extension is studied in~\cite{koiran-extensions}, where the following planted model is proposed in Section~5: For $i \in [m]$, let $Z_i = R^{-1} D_i R$ where $R$ is a generically chosen $r \times r$ matrix, and each $D_i$ is a generically chosen $r \times r$ diagonal matrix. For some $n \le r$, let $A_i$ be the upper-left $n \times n$ submatrix of $Z_i$. Given $(A_1,\ldots,A_m)$ and $r$, the goal is to construct a commuting extension $(\tilde{Z}_1,\ldots,\tilde{Z}_m)$ of dimension $r \times r$. The $Z_i$ make up one ``planted'' solution, but it is not unique: see Eq.~(2) of~\cite{koiran-extensions}.

The algorithm of~\cite{koiran-extensions} solves the above task provided $m \ge 3$ and $r \le 4n/3$. Building on the connection between commuting extensions and tensor decomposition described in~\cite{koiran}, we will give an algorithm for finding commuting extensions using our tensor decomposition algorithm. The algorithm will succeed for $r \le (2-\epsilon)n$ for an arbitrary $\epsilon > 0$, provided $m$ is a large enough constant (depending on $\epsilon$) and $n$ is large. The threshold $r=2n$ is fundamental: as pointed out in~\cite{koiran-commuting}, any tuple $(A_1,\ldots,A_m)$ of $n \times n$ matrices admits the commuting extension
\[ Z_i = \left(\begin{array}{cc} A_i & -A_i \\ A_i & -A_i \end{array}\right) \]
of dimension $r= 2n$.

The following lemma from~\cite{koiran} will be a subroutine in the algorithm.

\begin{lemma}[\cite{koiran}, Lemma~13]
\label{lem:extend}
Let $U \in \RR^{n \times r}$ and $V \in \RR^{r \times n}$ with $UV = I_n$. Then $r \ge n$, and it is possible to ``extend'' $U,V$ to two $r \times r$ matrices that are inverses. More precisely, $\tilde{U} = (U^\top \mid A)^\top$ and $\tilde{V} = (V \mid B)$ where $A,B \in \RR^{r \times (r-n)}$ and $\tilde{U} \tilde{V} = I_r$. Specifically, $B$ can be any matrix for which $\mathrm{Im}(B) = \ker(U)$, and $A$ is the unique matrix such that $A^\top B = I_{r-n}$ and $A^\top V = 0$.
\end{lemma}

\begin{algorithm}[Finding a commuting extension]
\label{alg:extension}
\phantom{}
\begin{itemize}
\item Input: A tuple of $n \times n$ matrices $(A_1,\ldots,A_m)$.
\item Input: A parameter $r \ge n$.
\item Input: A generic $n \times n$ matrix $M$ (say, chosen at random).
\end{itemize}
\begin{enumerate}
\item Let $T_i = MA_i$ for $i \in [m]$, and let $T$ be the $m \times n \times n$ tensor with slices $T_i$.
\item Let $q$ be the largest odd integer with $q \le m$, and let $p = (q-1)/2$.
\item Run the tensor decomposition algorithm (Algorithm~\ref{alg:decomp}) on $T$ with parameters $p,q,r$ to produce a decomposition $T = \sum_{\ell=1}^r a^{(\ell)} \otimes b^{(\ell)} \otimes c^{(\ell)}$.
\item Let $B$ be $n \times r$ with columns $b^{(\ell)}$, let $C$ be $r \times n$ with rows $c^{(\ell)\top}$, and for $i \in [m]$ let $\tilde{D}_i = \diag(a^{(\ell)}_i \,:\, \ell = 1,\ldots,r)$.
\item Solve for an $r \times r$ diagonal matrix $\hat{D}$ satisfying the linear equations $M^{-1} B \hat{D} C = I_n$.
\item Invoke Lemma~\ref{lem:extend} to extend $M^{-1} B \hat{D}$ and $C$ to $r \times r$ matrices $\tilde{U}$ and $\tilde{V}$ respectively.
\item Output $(\tilde{Z}_1,\ldots,\tilde{Z}_m)$ where $\tilde{Z}_i = \tilde{U} \hat{D}^{-1} \tilde{D}_i \tilde{V}$.
\end{enumerate}
\end{algorithm}

\noindent The matrix $M$ might be unnecessary, but it will help with the genericity analysis below.

\begin{theorem}
Let $r \ge n \ge 3$ and $m \ge 19$. The following holds for a generic choice of an $r \times r$ matrix $R$, an $n \times n$ matrix $M$, and $m$ diagonal $r \times r$ matrices $D_1,\ldots,D_m$ (that is, the $r^2 + n^2 + mr$ variables $R_{ij}, M_{ij}, (D_i)_{jj}$ are chosen generically from $\RR^{r^2 + n^2 + mr}$). For $i \in [m]$, let $Z_i = R^{-1} D_i R$ and let $A_i$ be the upper-left $n \times n$ submatrix of $Z_i$. If
\[ r \le 2n\left(1-\frac{4}{m-1}\right) - \frac{m^3}{4} \]
then Algorithm~\ref{alg:extension} with input $(A_1,\ldots,A_m)$, $r$, $M$ outputs a valid commuting extension $(\tilde{Z}_1,\ldots,\tilde{Z}_m)$, that is, the $\tilde{Z}_i$ pairwise commute and for $i \in [m]$, the upper-left $n \times n$ submatrix of $\tilde{Z}_i$ is $A_i$.
\end{theorem}

\begin{proof}
By construction, we know $A_i = U D_i V$ where $U$ is $n \times r$ consisting of the first $n$ rows of $R^{-1}$, and $V$ is $r \times n$ consisting of the first $n$ columns of $R$. This means $T_i = MU D_i V$, which gives a tensor decomposition $T = \sum_{\ell=1}^r w^{(\ell)} \otimes u^{(\ell)} \otimes v^{(\ell)}$ where $u^{(\ell)}$ are the columns of $MU$, $v^{(\ell)\top}$ are the rows of $V$, and $D_i = \diag(w^{(1)}_i,\ldots,w^{(r)}_i)$. We will verify that this decomposition of $T$ satisfies the conditions of our uniqueness theorem (Theorem~\ref{thm:uniqueness}), applied with the parameters $p,q$ specified in Algorithm~\ref{alg:extension}. The components $w^{(\ell)},u^{(\ell)},v^{(\ell)}$ are not generic in the usual sense, so we cannot apply Theorem~\ref{thm:generic-success} directly. Instead, since $R^{-1} = \adj(R)/\det(R)$, the entries of the components $w^{(\ell)},u^{(\ell)},v^{(\ell)}$ can be expressed as rational functions in the underlying variables
\[ x = \{R_{ij} \,:\, i,j \in [r]\} \sqcup \{(D_i)_{jj} \,:\, i \in [m], \, j \in [r]\} \sqcup \{M_{ij} \,:\, i,j \in [n]\}. \]
For generic $x$, our goal is to verify conditions~\ref{item:a1}-\ref{item:P-prime} from Theorem~\ref{thm:uniqueness} for the components $w^{(\ell)},u^{(\ell)},v^{(\ell)}$. Inspecting each of these conditions and recalling that the components are rational functions, it suffices to exhibit a single choice of the underlying variables $x$ (with $R,M$ invertible) for which conditions~\ref{item:a1}-\ref{item:P-prime} hold; it then follows that they hold for generic $x$. From Theorem~\ref{thm:generic-success} we know that a generic choice of components will satisfy the conditions (we will verify the requirement of the theorem later); fix such a generic choice $\tilde{w}^{(\ell)},\tilde{u}^{(\ell)},\tilde{v}^{(\ell)}$. We will exhibit values for the variables $x = (R,D_i,M)$ such that $w^{(\ell)}(x) = \tilde{w}^{(\ell)}$, $u^{(\ell)}(x) = \tilde{u}^{(\ell)}$, and $v^{(\ell)}(x) = \tilde{v}^{(\ell)}$, implying that $w^{(\ell)},u^{(\ell)},v^{(\ell)}$ satisfy the conditions as desired.

First, set $D_i = \diag(\tilde{w}^{(1)}_i,\ldots,\tilde{w}^{(r)}_i)$, which ensures $w^{(\ell)}(x) = \tilde{w}^{(\ell)}$. To ensure $u^{(\ell)}(x) = \tilde{u}^{(\ell)}$, we need to arrange $MU = \tilde{U}$ where $\tilde{U}$ has columns $\tilde{u}^{(1)},\ldots,\tilde{u}^{(r)}$ and where $U$ consists of the first $n$ rows of $R^{-1}$. Equivalently, $U = M^{-1} \tilde{U}$. Similarly, to ensure $v^{(\ell)}(x) = \tilde{v}^{(\ell)}$, we need to arrange $V = \tilde{V}$ where $\tilde{V}$ has rows $\tilde{v}^{(1)\top},\ldots,\tilde{v}^{(r)\top}$ and where $V$ consists of the first $n$ columns of $R$. Set $M = \tilde{U} \tilde{V}$, which is invertible due to the generic choice of $\tilde{u}^{(\ell)}, \tilde{v}^{(\ell)}$. For our desired values of $U = M^{-1} \tilde{U}$ and $V = \tilde{V}$, we have $UV = M^{-1} \tilde{U} \tilde{V} = I_n$. Therefore by Lemma~\ref{lem:extend} there exists an $r \times r$ matrix $R$ such that the first $n$ rows of $R^{-1}$ are $U = M^{-1} \tilde{U}$ and the first $n$ columns of $R$ are $V = \tilde{V}$. This completes the construction of $x = (R,D_i,M)$ and therefore completes the proof that the components $w^{(\ell)},u^{(\ell)},v^{(\ell)}$ satisfy the conditions of Theorem~\ref{thm:uniqueness} when $R,D_i,M$ are chosen generically.

Now by Theorem~\ref{thm:uniqueness} we know $T$ has a unique (in the sense of Definition~\ref{def:unique}) rank-$r$ decomposition, and that the components $a^{(\ell)},b^{(\ell)},c^{(\ell)}$ recovered by our algorithm are related to $w^{(\ell)},u^{(\ell)},v^{(\ell)}$ as follows: there exists a permutation $\pi: [r] \to [r]$ such that $a^{(\ell)} \otimes b^{(\ell)} \otimes c^{(\ell)} = w^{(\pi(\ell))} \otimes u^{(\pi(\ell))} \otimes v^{(\pi(\ell))}$ for all $\ell \in [r]$. This means $b^{(\ell)} = \beta_\ell \cdot u^{(\pi(\ell))}$ for a scalar $\beta_\ell \ne 0$, which we can write in matrix form as $B = MU \Pi \cdot \diag(\beta)$ where $\Pi$ is the permutation matrix associated with $\pi$ and $B$ is defined in Algorithm~\ref{alg:extension}. Similarly, $c^{(\ell)} = \gamma_\ell \cdot v^{(\pi(\ell))}$ for a scalar $\gamma_\ell \ne 0$, and $C = \diag(\gamma) \cdot \Pi^\top V$. The algorithm's next step is to solve for a diagonal matrix $\hat{D}$ such that $M^{-1} B \hat{D} C = I_n$. One solution to this linear system is $\hat{D} = \diag(\beta_\ell^{-1} \gamma_\ell^{-1} \,:\, \ell = 1,\ldots,r)$; to see this, note that $UV = I_n$ (since $U,V$ are defined as particular submatrices of $R^{-1},R$) and $\Pi \cdot \diag(\beta) \cdot \hat{D} \cdot \diag(\gamma) \cdot \Pi^\top = \diag(\beta_{\pi^{-1}(\ell)} \hat{D}_{\pi^{-1}(\ell)} \gamma_{\pi^{-1}(\ell)} \,:\, \ell = 1,\ldots,r)$. Later, we will show that generically, this is the only solution. Therefore the algorithm recovers $\hat{D} = \diag(\beta_\ell^{-1} \gamma_\ell^{-1} \,:\, \ell = 1,\ldots,r)$. The algorithm's next step is to extend $M^{-1} B \hat{D}$ and $C$ to $r \times r$ matrices $\tilde{U},\tilde{V}$ with $\tilde{U} \tilde{V} = I_r$, where $M^{-1} B \hat{D}$ makes up the first $n$ rows of $\tilde{U}$ and $C$ makes up the first $n$ columns of $\tilde{V}$; this uses Lemma~\ref{lem:extend} along with the previous claim $M^{-1} B \hat{D} C = I_n$.

Recall that the algorithm has access to
\[ \tilde{D}_i = \diag(a^{(\ell)}_i \,:\, \ell = 1,\ldots,r) = \diag(\beta_\ell^{-1} \gamma_\ell^{-1} w^{(\pi(\ell))}_i \,:\, \ell = 1,\ldots,r) \]
and outputs $\tilde{Z}_i = \tilde{U} \hat{D}^{-1} \tilde{D}_i \tilde{V}$ for $i \in [m]$. We need to show this is a valid commuting extension. Since $\tilde{U},\tilde{V}$ are inverses and $\hat{D}^{-1} \tilde{D}_i$ is diagonal, the $Z_i$ pairwise commute. To verify that $Z_i$ ``extends'' $A_i$: the upper-left $n \times n$ block of $Z_i$ is
\begin{align*}
M^{-1}B\hat{D} \cdot \hat{D}^{-1} \tilde{D}_i \cdot C &= M^{-1} (MU \Pi \cdot \diag(\beta)) \tilde{D}_i (\diag(\gamma) \cdot \Pi^\top V) \\
&= U \Pi \cdot \diag(\beta) \cdot \tilde{D}_i \cdot \diag(\gamma) \cdot \Pi^\top V \\
&= U \Pi \cdot \diag(w_i^{(\pi(\ell))} \,:\, \ell = 1,\ldots,r) \cdot \Pi^\top V \\
&= U \cdot \diag(w_i^{(\ell)} \,:\, \ell = 1,\ldots,r) \cdot V \\
&= U D_i V \\
&= A_i
\end{align*}
as desired.

Next we will prove the claim, made above, that the linear system of equations for $\hat{D}$ has a unique solution. The equations are $M^{-1} B \hat{D} C = I_n$ where $\hat{D}$ is diagonal, so it suffices to show that the matrices $M^{-1} b^{(\ell)} c^{(\ell)\top}$ for $\ell \in [r]$ are linearly independent. Since $M^{-1}$ is invertible, it suffices to show that the matrices $b^{(\ell)} c^{(\ell)\top}$ for $\ell \in [r]$ are linearly independent. Since $b^{(\ell)}$ and $c^{(\ell)}$ are nonzero scalar multiples of $u^{(\pi(\ell))}$ and $v^{(\pi(\ell))}$ respectively, it suffices to show that $u^{(\ell)} v^{(\ell)\top}$ for $\ell \in [r]$ are linearly independent. Recalling that $u^{(\ell)}$ are the columns of $MU$ and $M$ is invertible, it suffices to show that $\hat{u}^{(\ell)} v^{(\ell)\top}$ for $\ell \in [r]$ are linearly independent, where $\hat{u}^{(\ell)}$ are the columns of $U$. Recalling that the entries of $R^{-1}$ are rational functions in the entries of $R$, it suffices to exhibit a single (invertible) value of $R$ for which $\hat{u}^{(\ell)} v^{(\ell)\top}$ are linearly independent, where $\hat{u}^{(\ell)}$ and $v^{(\ell)}$ are viewed as functions of $R$: namely, the first $n$ rows of $R^{-1}$ have columns $\hat{u}^{(\ell)}$ and the first $n$ columns of $R$ have rows $v^{(\ell)\top}$. Constructing such a value for $R$ amounts to finding linearly independent matrices $\hat{u}^{(\ell)} v^{(\ell)\top}$ for $\ell \in [r]$ that sum to $I_n$, since the rest of $R^{-1}, R$ can be filled in by Lemma~\ref{lem:extend}. This in turn is guaranteed by Lemma~\ref{lem:rank-1-indep} below, noting that our assumption $n \ge 3$ implies $n^2-n+1 \ge 2n$.

Finally, we need to check that our parameters $p,q,r$ verify the requirements for Theorem~\ref{thm:generic-success}. Note we are in the case $\alpha = 1$ so we need $q \ge 18$ and $r \le 2n(1-4/q) - q^3/4$. These follow from our assumptions, and our choice $m-1 \le q \le m$.
\end{proof}

\begin{lemma}\label{lem:rank-1-indep}
For $n \ge 1$ and $n \le r \le n^2-n+1$, there exist $r$ linearly independent rank-1 $n \times n$ matrices that sum to $I_n$.
\end{lemma}

\begin{proof}
The rank-1 matrices used in our construction will each be supported on a rectangle, either with all nonzero values equal to $1$ or all nonzero values equal to $-1$. One possible construction is to take the rectangle $[n] \times [n]$ and subtract off a collection of rectangles that partition $([n] \times [n]) \setminus \{(1,1),(2,2),\ldots,(n,n)\}$. The number of rectangles in the partition can range from $2(n-1)$ to $n^2-n$, so this covers all the $r$ values ranging from $2(n-1)+1$ to $n^2-n+1$. An alternative construction is to, for some $k \in [n]$, take the rectangle $[k] \times [k]$ and subtract off a collection of $2(k-1)$ rectangles that partition $([k] \times [k]) \setminus \{(1,1),(2,2),\ldots,(k,k)\}$, and then add $n-k$ additional $1 \times 1$ rectangles on the diagonal. This uses $1+2(k-1)+(n-k) = n+k-1$ rectangles in total, so by varying $k$ this covers all the $r$ values ranging from $n$ to $2n-1$.
\end{proof}

\addcontentsline{toc}{section}{Acknowledgments}
\section*{Acknowledgments}

We thank Aravindan Vijayaraghavan for helpful discussions.

\addcontentsline{toc}{section}{References}
\bibliographystyle{alpha}
\bibliography{main}

\end{document}